\documentclass[11pt]{article}
\usepackage[T1]{fontenc}
\usepackage[utf8]{inputenc}
\usepackage{lmodern,microtype}
\usepackage{geometry,setspace,graphicx,xcolor}
\usepackage{amsmath,amssymb,amsthm,mathtools}
\usepackage{booktabs,array,threeparttable,caption,subcaption,float}
\usepackage{natbib,enumerate,comment,etoolbox}
\usepackage{hyperref}
\newtheorem{lemma}{Lemma}

\newtheorem{corollary}{Corollary}
\newtheorem{proposition}{Proposition}

\theoremstyle{definition}

\newtheorem{remark}{Remark}

\newtheorem{example}{Example}
\newtheorem*{example*}{Example}

\newtheorem*{fact*}{Fact}

\includecomment{proofs}

\hypersetup{colorlinks,citecolor=blue} 
\title{Demand Curvature and Pass-Through in Multiproduct Oligopoly\footnote{I thank Jinhyuk Lee, Aviv Nevo, Devesh Raval, and the participants of the 2026 KASIO Summer Conference and Korean Econometric Society Summer Conference for helpful comments. All errors are mine.}}
\author{Paul S. Koh\footnote{School of Economics, Yonsei University. Email: \texttt{paulkoh9@gmail.com}.} }
\date{September 13, 2026}

\begin{document}

\maketitle
\begin{abstract}
Economic interventions change firms' pricing incentives, but their effects depend on how those incentives propagate across products and firms. This paper develops tractable characterizations of that propagation under multiproduct Bertrand competition. A decomposition of the pricing system isolates demand curvature, substitution, and ownership, linking the equilibrium response matrix to commonly used empirical demand models. It yields approximations with explicit error bounds and separates individual adjustment from equilibrium feedback. Small-share limits reveal when interactions disappear and when substitution within nests or selection among heterogeneous consumers preserves them. For nested logit, a closed-form response reduces the limiting product-level system to averages within firm--nest groups, exposing the direction of price spillovers. The same framework organizes local responses to changes in costs, demand, and ownership. The results clarify which features of demand support simple incidence predictions and which interactions those predictions must retain.

\vspace{1em}
\end{abstract}

\clearpage
\section{Introduction}

Economic interventions change the incentives to set prices. Their consequences depend on how those incentives propagate through a market: a firm may reprice other products in its portfolio, competitors may respond, and consumers may substitute toward products that the intervention does not directly affect. This propagation is common to otherwise different questions about market outcomes and welfare. Pass-through provides a way to study it by separating a change in pricing incentives from the equilibrium response to that change \citep{weyl2013pass}. The challenge in multiproduct markets is to make the resulting response matrix tractable without discarding the interactions that matter for the question at hand.

Demand elasticity alone does not answer this question. In simulated cereal data, \citet{miravete2023pass} obtain similar average elasticities and relative markups across estimated demand specifications, but substantially different pass-through predictions. A firm's response to a cost change depends on how a price adjustment changes its marginal incentive to raise price further. For a monopolist selling one product, demand curvature governs this adjustment \citep{bulow1983note}. With differentiated products, changing one price also changes the sensitivity of demand to other prices. These mixed derivatives affect the value of sales diverted within a portfolio and the incentives of competitors to reprice. Ownership determines which effects a firm internalizes; equilibrium combines the resulting adjustments. Thus, even detailed knowledge of substitution at observed prices leaves a substantive role for assumptions about how substitution changes as prices move.

This paper develops a tractable characterization of local price responses under multiproduct Bertrand competition. The characterization links the equilibrium response matrix to own-price curvature, changes in substitution, and product ownership. It supplies two ways to simplify that matrix: a sequence of approximations with explicit error control, and analytical limits tied to widely used demand specifications. The objective is to make first-order analysis easier to implement and interpret in these models. A firm's response with rival prices fixed and its response after rivals adjust are distinct objects; likewise, small unconditional shares and weak substitution among the remaining purchases are distinct restrictions. Keeping these distinctions explicit identifies which interactions a useful simplification must retain.

The starting point is a decomposition of the pricing system. Normalizing each first-order condition by its own-price demand slope separates the price derivative of that condition into a baseline own-price term, a demand-curvature and substitution term, and an ownership term. The ownership term records both the internalization of diverted sales and changes in diversion as prices adjust. Inverting this system gives the equilibrium pass-through matrix. The decomposition therefore traces how demand shape and product portfolios jointly transmit a pricing incentive, including its effects on products not directly affected by the intervention. In mixed logit, these derivatives reflect both individual choice responses and changes in the composition of buyers. Cost shocks provide the initial representation; a general intervention enters through a different change in the pricing conditions.

The distinction between individual adjustment and equilibrium feedback has a direct economic implication. A firm's profit curvature characterizes its local adjustment to a cost change when rival prices are fixed, but does not by itself determine the equilibrium response. A quadratic-utility duopoly illustrates the difference at a Nash equilibrium: fixed-rival own-cost pass-through is one-half, whereas a common cost increase passes through at four-fifths. Both firms choose globally optimal prices, and the common-cost response holds exactly over a finite range of cost changes. Across the example's demand family, baseline quantities, markups, and own-profit curvature remain unchanged while equilibrium pass-through varies with substitution. The example illustrates established strategic feedback; the decomposition identifies how that feedback enters the response matrix and its approximations.

The decomposition also yields a sequence of approximations that adds successive interactions to a diagonal baseline. A Neumann expansion supplies an exact remainder and an error bound under an explicit contraction condition. For common-coefficient multinomial logit, I establish convergence at every interior equilibrium, allowing multiproduct ownership. In a numerical exercise with 1,000 logit configurations with heterogeneous product characteristics and costs, the first interaction correction has a mean relative matrix error of $0.323\%$. This exercise measures approximation of the local response; separate finite-cost experiments assess the additional error from using a local derivative to predict a discrete change. The analytical error bounds give sufficient conditions under which the approximation is justified beyond this logit benchmark. A residual bound also evaluates a proposed response to a particular shock and, when sufficiently tight, certifies the sign of a local price or price-index effect.

Small-share limits reveal a different source of simplification and its limits. Shares may fall because the outside option becomes more attractive, prices rise, or demand approaches a boundary. These experiments need not leave the same consumers or substitution patterns in the market. Along the specified fixed-dimensional paths, logit and lognormal mixed logit yield identity unit-pass-through limits, whereas CES yields a diagonal limit above one. Cross-product effects can nevertheless survive when every product's unconditional share vanishes. In nested logit, substitution and ownership within a disappearing nest continue to determine the limiting response; gamma mixing can preserve interaction along a common high-price ray. Small shares alone therefore do not justify treating products as independent for incidence analysis. The economically relevant question is whether substitution among the remaining purchases disappears with them. In nested logit, the surviving response admits a closed form based on cost averages within firm--nest groups. It reveals nonnegative cross-firm and nonpositive within-firm cross-price responses in the limiting nest, while a common cost change passes through one for one. This reduction preserves economically relevant spillovers without requiring a product-level matrix inversion.

The applications show how the same response system can be used across economic questions. A general incentive formula covers smooth changes in costs, demand, and ownership. A demand intervention affects pricing through both its effect on sales and its effect on price sensitivities; observing the first effect alone is generally insufficient. A demand-translation example shows how changes in willingness to pay map into prices through the cost pass-through matrix. The remaining applications distinguish unit from percentage incidence, derive revenue-weighted consumer-surplus expressions under the small-share benchmarks, and connect ownership incentives to local merger predictions. These results link the general first-order approach to explicit calculations in empirical demand models, conditional on the demand derivatives and pricing behavior used to evaluate them.

\subsection{Related Literature}

The paper builds on the theory of incidence under imperfect competition. \citet{bulow1983note} establish the role of demand curvature in monopoly pass-through. \citet{weyl2013pass} relate pass-through to conduct and welfare, developing both symmetric and asymmetric analyses, with a matrix characterization in Appendix C. \citet{adachi2022pass} characterize welfare incidence for multiple policy instruments and heterogeneous firms. On a given one-dimensional demand curve, the Bulow--Pfleiderer index used here and Weyl and Fabinger's marginal-consumer-surplus elasticity encode the same curvature information. The contribution here concerns tractability within multiproduct Bertrand models: the decomposition makes the derivative requirements explicit, the approximation bounds control omitted interactions, and the model-specific limits identify when those interactions simplify. The general use of pass-through for interventions is the starting point of the analysis.

The retail literature studies curvature-dependent responses to manufacturer trade deals \citep{tyagi1999characterization} and own-brand and cross-brand pass-through \citep{besanko2005own}. \citet{moorthy2005general} derives a matrix characterization for competing multiproduct retailers and separates cost-induced substitution incentives from internal and external strategic complementarity. His explicit nested-logit solution distinguishes product-, retailer-, and marketwide shocks: a cost increase on one product at one retailer lowers the other product's price at that retailer and raises both prices at its competitor. These mechanisms are antecedents of the analysis here. The present decomposition identifies their demand-curvature and ownership components and studies approximations and limiting responses under specified demand systems.

\citet{alexandrov2017lechatelier} analyze the LeChatelier--Samuelson principle in games and conditions under which additional adjustments amplify responses. Pass-through also has an established role in merger analysis. \citet{froeb2005pass} connect the price effects of mergers to the pass-through of efficiencies; \citet{jaffe2013first} develop local price and welfare approximations using merger pass-through. \citet{miller2016pass, miller2017upward} evaluate related prediction methods. The present analysis exposes the interaction terms omitted by low-order approximations and connects them to demand shape and ownership in empirical models. The merger application retains Jaffe and Weyl's distinction between merger pass-through and ordinary premerger cost pass-through.

A closely related literature examines demand shape and counterfactual sensitivity. \citet{mrazova2017not} organize firm behavior through elasticity and convexity. \citet{miravete2023pass} connect demand specification to differentiated-product incidence, and \citet{miravete2026elasticity} explain how preference heterogeneity shapes curvature and develop a more flexible demand specification. Their demand-derivative calculations and selection mechanisms are antecedents of the mixed-logit interpretation below. The present analysis connects those derivatives to the multiproduct pricing Jacobian and gives conditions for assessing omitted interactions in a particular local response. It thereby helps interpret specification-sensitive incidence predictions without attributing their quantitative differences to a single channel. \citet{birchall2024estimating} use a Box--Cox transformation to relax restrictions on substitution and curvature, while \citet{compiani2022market} estimates multiproduct demand nonparametrically and examines sensitivity across tax and product-portfolio counterfactuals. These contributions motivate the focus here on the derivative combinations entering equilibrium pricing and the restrictions under which they simplify.

\citet{armstrong2018multiproduct} derive tractable multiproduct pricing and pass-through results for monopoly and Cournot competition. \citet{armstrong2023multiproduct} characterize feasible multiproduct-monopoly pass-through matrices and show how price responses can have counterintuitive signs. The present analysis studies asymmetric multiproduct Bertrand competition with constant marginal costs, developing approximations and demand-specific limits for its local response system.

\citet{nocke2018multiproduct, nocke2025aggregative} exploit aggregative structure to analyze multiproduct oligopoly and mergers, including logit and CES demand. Their nested extension allows firms that own entire nests and firms confined to part of a single nest. \citet{garrido2024aggregative} allows arbitrary overlap between firms and nests and develops a reduced pricing system using common markups within firm--nest groups. The nested calculation here studies a specified vanishing-share limit of this pricing structure and expresses the resulting local cost response through group averages, making its spillover signs explicit. The general local decomposition does not require aggregative demand.

\citet{sangani2026complete} documents complete pass-through in levels for common cost shocks and studies demand restrictions underlying pass-through in levels and percentages. This comparison is especially relevant to the small-share and percentage-incidence results. A common-shock response restricts a weighted sum of columns of the pass-through matrix; it need not determine the response to a shock affecting only some products. The present framework retains this distinction when deriving model-specific limits and applying them to incidence.

Finally, \citet{mackay2014bias} study bias in reduced-form pass-through estimates, and \citet{dearing2024learning} connect the relevance of instruments for conduct tests to the features of pass-through those instruments target, including their relevance for counterfactuals. Their work addresses what observed variation reveals about pricing behavior. The present results are conditional comparative statics: they specify which demand derivatives and equilibrium interactions a prediction requires. Identifying those objects from observed shocks is a separate empirical task.

Section~\ref{section:setup} introduces the model. Section~\ref{section:characterization.of.pass.through} characterizes pass-through and separates fixed-rival responses from equilibrium feedback. Sections~\ref{section:finite.series.approximation} and~\ref{section:small.share.approximation} develop the series approximation and small-share limits. Section~\ref{section:applications} presents the applications, and Section~\ref{section:conclusion} concludes. The appendices contain proofs, model-specific derivations, numerical illustrations, and supplementary directional bounds.

\section{Setup \label{section:setup}}

\subsection{Model}

Consider a Bertrand--Nash pricing game among multiproduct firms selling differentiated products. Let $\mathcal{J}=\{1,\dots,J\}$ denote the set of inside goods, and let $j=0$ denote the outside option. Firm $f$ produces the subset $\mathcal{J}_f \subseteq \mathcal{J}$ and earns profit
\begin{equation}\label{equation:firm.profit.function}
\Pi_f(\mathbf p)
=
\sum_{l \in \mathcal{J}_f} (p_l - c_l)\, q_l(\mathbf p),
\end{equation}
where $p_l$, $c_l$, and $q_l(\mathbf p)$ denote price, constant marginal cost, and demand for product $l$.

The analysis is local to an interior Bertrand--Nash equilibrium that admits an equilibrium continuation after sufficiently small cost changes. Demand is twice continuously differentiable in a neighborhood of this point, with $q_j>0$ and $\partial q_j/\partial p_j<0$. The full Jacobian of the pricing first-order conditions is nonsingular, so the implicit function theorem makes this continuation locally unique and differentiable among solutions to the pricing conditions. These assumptions do not require global uniqueness; first-order conditions alone need not be sufficient for equilibrium. The general characterization accommodates smooth demand with heterogeneous product tastes and price sensitivities; the distributional restrictions in the model-specific limits are additional assumptions. At the equilibrium, each firm's first-order conditions require, for each $j \in \mathcal{J}_f$,
\begin{equation}\label{equation:firm.foc}
\frac{\partial \Pi_f}{\partial p_j}
=
q_j(\mathbf p)
+
\sum_{l \in \mathcal{J}} \Omega_{jl}(p_l - c_l)\frac{\partial q_l(\mathbf p)}{\partial p_j}
=
0,
\end{equation}
where $\Omega$ is the symmetric $J\times J$ ownership matrix with $\Omega_{jl}=1$ if products $j$ and $l$ are owned by the same firm and $0$ otherwise. Stacking \eqref{equation:firm.foc} over $j \in \mathcal{J}_f$ yields
\begin{equation}\label{equation:firm.foc.stacked}
\nabla_{\mathbf p_f}\Pi_f(\mathbf p)
=
\mathbf q_f(\mathbf p)
+
\mathbf{J}_{\mathbf{q}_f}(\mathbf p)^\top (\mathbf p_f - \mathbf c_f)
=
\mathbf 0,
\end{equation}
where $\mathbf{q}_f(\mathbf{p})$ is the vector of demands for firm $f$'s products, and $\mathbf{J}_{\mathbf{q}_f}(\mathbf p)\equiv \frac{\partial \mathbf q_f(\mathbf p)}{\partial \mathbf p_f^\top}$ is the within-firm Jacobian with typical element $\mathbf{J}_{\mathbf{q}_f,kj}(\mathbf p)=\frac{\partial q_k(\mathbf p)}{\partial p_j}$.

Stacking \eqref{equation:firm.foc} across all firms gives the system
\begin{equation}\label{equation:firm.foc.stacked.across.firms}
\mathbf F(\mathbf p)
\equiv
\mathbf q(\mathbf p)
+
\Delta(\mathbf p) (\mathbf p - \mathbf c)
=
\mathbf 0,
\end{equation}
where $\Delta(\mathbf p) \equiv \Omega \circ \mathbf{J}_{\mathbf{q}}(\mathbf p)^\top$ ($\circ$ denotes the Hadamard product) and $\mathbf{J}_{\mathbf{q}}(\mathbf p)\equiv \frac{\partial \mathbf q(\mathbf p)}{\partial \mathbf p^\top}$. All pass-through calculations below evaluate this system at the equilibrium under study. When discussing a sequence of equilibria, marginal costs may vary along the sequence; an arbitrary path of prices is not itself an equilibrium path.

\subsection{Definition of the Pass-Through Matrix}

I now introduce the main object of interest, the pass-through matrix, which describes how per-unit taxes or cost shocks affect equilibrium prices. Recall that the firms' pricing equations are
\[
\mathbf{F}(\mathbf{p}) = \mathbf{0} \quad \Leftrightarrow \quad q_j + (p_j - c_j) \frac{\partial q_j}{\partial p_j} +\sum_{l \in \mathcal{J} \backslash \{j\}} \Omega_{jl}(p_l - c_l) \frac{\partial q_l}{\partial p_j} = 0, \;\; \forall j \in \mathcal{J}.
\]
To derive the pass-through matrix, it is convenient to rewrite the first-order conditions in a form that is quasilinear in marginal costs:
\begin{equation}\label{equation:foc.quasilinear.in.mc.stacked}
\mathbf f(\mathbf p)
 \equiv
-\bigl(\mathbf J_{\mathbf q}^{\mathrm{diag}}(\mathbf p)\bigr)^{-1}\mathbf F(\mathbf p)
=
\mathbf 0 \quad \Leftrightarrow \quad
- \frac{q_j}{\partial q_j / \partial p_j}
-(p_j-c_j)
+\sum_{l\in\mathcal J\setminus\{j\}}\Omega_{jl}(p_l-c_l)
D_{j \to l} = 0, \; \forall j \in \mathcal{J},
\end{equation}
where \(\mathbf J_{\mathbf q}^{\mathrm{diag}}(\mathbf p)\) is a \(J\times J\) diagonal matrix with \((j,j)\) entry \(\partial q_j(\mathbf p)/\partial p_j\), and $D_{j \to l} \equiv - \frac{\partial q_l / \partial p_j}{\partial q_j / \partial p_j}$ is the diversion ratio from product $j$ to $l$.

Now consider a vector of per-unit taxes \(\mathbf t=(t_1,\dots,t_J)\) so that product \(j\)'s marginal cost becomes \(c_j+t_j\). The post-tax equilibrium satisfies
\begin{equation}\label{equation:post.tax.foc}
\mathbf f(\mathbf p)+\Lambda(\mathbf{p}) \mathbf t=\mathbf 0,
\end{equation}
where $D_{j\to j}=-1$, so \(\mathbf D=[D_{j\to k}]_{j,k\in\mathcal J}\) includes this diagonal convention and
\(\Lambda(\mathbf{p})\equiv -\Omega\circ \mathbf D(\mathbf{p})\) is the ownership-internalized diversion matrix.\footnote{Equivalently, $\Lambda_{jl} = \Omega_{jl}\frac{\partial q_l/\partial p_j}{\partial q_j/\partial p_j} = -\Omega_{jl}D_{j\to l}$. If firms are single-product producers, then \(\Lambda=I\).} Differentiating \eqref{equation:post.tax.foc} with respect to \(\mathbf t^\top\) gives
\[
\frac{\partial \mathbf{f}(\mathbf{p})}{\partial \mathbf{p}^\top} \frac{\partial \mathbf{p}}{\partial \mathbf{t}^\top} + \Lambda(\mathbf{p}) + \sum_{m=1}^J \left( \frac{\partial \Lambda(\mathbf{p})}{\partial p_m} \mathbf{t} \right) \left( \frac{\partial p_m(\mathbf{t})}{\partial \mathbf{t}^\top} \right) = \mathbf{0}.
\]
Evaluating the expression at $\mathbf{t} = \mathbf{0}$ and rearranging (assuming that $\mathbf{J}_{\mathbf{f}}$ is invertible) yields the \emph{pass-through matrix}
\begin{equation}\label{equation:pass.through.matrix.definition}
\frac{\partial \mathbf p}{\partial \mathbf t^\top}
=
-
\mathbf{J}_{\mathbf{f}}^{-1}\Lambda,
\end{equation}
where $\mathbf{J}_{\mathbf{f}} \equiv \frac{\partial \mathbf f(\mathbf p)}{\partial \mathbf p^\top}$. Write this matrix as $\Psi$.\footnote{Up to row normalization and the choice of cost-shock columns, this is the comparative-static system $P=-H^{-1}C$ in \citet[equation (2)]{moorthy2005general}. His $H$ is the Jacobian of the unnormalized pricing conditions and his $C$ records their derivatives with respect to the specified cost changes. Normalization changes the representation, not the equilibrium response.} Its entry $\Psi_{jk}$ is the local equilibrium price response of product $j$ to a unit cost increase on product $k$, allowing all firms to reprice. Because $\mathbf f$ depends on first derivatives of demand, $\Psi$ generally requires second derivatives as well. Existing taxes can be included in baseline costs; the formula is then evaluated at zero \emph{incremental} tax. A finite policy change requires following the equilibrium branch or resolving the post-policy equilibrium.

\begin{remark}[The direction of the cost shock]
A local cost change $\mathbf t=\tau\mathbf z$ induces $d\mathbf p/d\tau=\Psi\mathbf z$. For a uniform per-unit shock, $\mathbf z=\mathbf 1$, so each price response is a row sum of $\Psi$. Complete pass-through of this common shock means $\Psi\mathbf 1=\mathbf 1$; this condition alone does not imply $\Psi=I$. For a proportional cost change $d\mathbf c=\mathbf c\,d\tau$ at positive prices and costs, the vector of percentage price responses is $\operatorname{diag}(\mathbf p)^{-1}\Psi\mathbf c$. Thus both the units of measurement and the direction of the shock matter when comparing pass-through restrictions across demand systems \citep{moorthy2005general,sangani2026complete}.
\end{remark}

\subsection{The Monopoly Curvature Benchmark}
For a single product with downward-sloping residual demand $q(p)$ and constant marginal cost, the first-order condition gives $p-c=-q/q'$. At this stationary point,
\[
\pi''(p)=2q'(p)+(p-c)q''(p)=q'(p)(2-\kappa(p)),
\]
where the Bulow--Pfleiderer index is
\begin{equation}\label{equation:single.product.soc.curvature.constraint}
\kappa(p):=\frac{q(p)q''(p)}{q'(p)^2}\leq 2.
\end{equation}
The inequality is the second-order necessary condition. Under the strict condition $\kappa<2$, differentiating the pricing equation gives\footnote{For the same downward-sloping demand curve, this is also the formula in \citet{weyl2013pass}. Let $p=P(q)$ be inverse demand and let $ms(q)=-qP'(q)$ denote marginal consumer surplus. Under their inverse-elasticity convention, $1/\epsilon_{ms}=q\,ms'(q)/ms(q)=1+qP''(q)/P'(q)=1-\kappa_{\mathrm{BP}}$, where $\kappa_{\mathrm{BP}}=q(p)q''(p)/q'(p)^2$. Thus $1/(1+1/\epsilon_{ms})=1/(2-\kappa_{\mathrm{BP}})$. In oligopoly, the common-price demand path used in a symmetric incidence calculation generally differs from an own-price demand curve holding other prices fixed. The scalar identity applies to the same curve; it does not equate curvature across these paths.}
\begin{equation}\label{equation:pass.though.in.single.product.monopoly}
\frac{\partial p}{\partial t}=\frac{1}{2-\kappa}.
\end{equation}

More convex demand makes the marginal profit response to a price increase less negative, increasing cost pass-through. In a multiproduct firm, the relevant profit curvature depends on the direction of joint price adjustment. In oligopoly, rival responses add a further source of amplification. The next section separates these two mechanisms.

\section{Characterization of Pass-Through \label{section:characterization.of.pass.through}}

\subsection{A Structural Decomposition of Pass-Through}
A product's pricing equation balances its own lost sales against margins and sales recaptured elsewhere in the firm's portfolio. Differentiating that equation separates two sources of change: the curvature of residual demand, and the way ownership makes diversion to other products valuable. The following decomposition keeps both sources visible before the full pricing system is inverted. Its purpose is to expose which parts of an empirical demand specification govern the response: fitting substitution at a price vector disciplines the slopes, while the mixed derivatives determine how pricing incentives change away from that vector.

Use $q_{j,k}=\partial q_j/\partial p_k$ and $q_{j,kl}=\partial^2q_j/\partial p_k\partial p_l$. Wherever the relevant first derivatives are nonzero, define
\begin{equation}\label{equation:demand.curvature.measure}
\kappa^j_{kl}:=\frac{q_jq_{j,kl}}{q_{j,k}q_{j,l}}.
\end{equation}
This index measures changes in a demand slope relative to the corresponding first-order effects. Its own-price case is the Bulow--Pfleiderer index.\footnote{When the denominators exist, $\kappa^j_{kl}=1+\partial_{kl}\log q_j/(\partial_k\log q_j\,\partial_l\log q_j)$. Directional curvature can then be written as $\kappa_j(v)=\sum_{k,l}\omega_{jkl}(v)\kappa^j_{kl}$, where $\omega_{jkl}(v)=(v_kq_{j,k})(v_lq_{j,l})/(\sum_m v_mq_{j,m})^2$. The weights sum to one but can be negative, so this is an affine combination, not necessarily an average with positive weights.} The matrix characterization itself does not require nonzero cross-price slopes.

\begin{lemma}[Decomposition of the FOC Jacobian]
\label{lemma:matrix.curvature.decomposition}
For twice continuously differentiable demand with $q_j>0$ and $q_{j,j}<0$, the normalized pricing Jacobian satisfies
\begin{equation}\label{equation:normalized.foc.jacobian.decomposition}
\mathbf J_{\mathbf f}=-2I+\mathbf K+\mathbf C,
\end{equation}
where
\begin{gather}
K_{jk}=\begin{cases}
q_jq_{j,jj}/q_{j,j}^2,&j=k,\\[2pt]
-q_{j,k}/q_{j,j}+q_jq_{j,jk}/q_{j,j}^2,&j\ne k,
\end{cases}\label{equation:K.matrix.elements}\\[4pt]
C_{jk}=\begin{cases}
\displaystyle\sum_{l\ne j}\Omega_{jl}(p_l-c_l)\frac{\partial D_{j\to l}}{\partial p_j},&j=k,\\[4pt]
\displaystyle\Omega_{jk}D_{j\to k}+\sum_{l\ne j}\Omega_{jl}(p_l-c_l)\frac{\partial D_{j\to l}}{\partial p_k},&j\ne k.
\end{cases}\label{equation:C.matrix.elements}
\end{gather}
When $q_{j,k}\ne0$, the off-diagonal entry also equals $K_{jk}=\delta_{jk}(1-\kappa^j_{kj})$, with $\delta_{jk}=-q_{j,k}/q_{j,j}$.
\end{lemma}
\begin{proof}
See Appendix \ref{section:proof.of.lemma.2}.
\end{proof}
The baseline $-2I$ comes from differentiating own demand divided by its own slope and the own markup. The matrix $\mathbf K$ is the demand-side correction: its diagonal measures own-price curvature; its off-diagonal entries combine the effect of another price on demand with its effect on the own-demand slope. The expanded definition in \eqref{equation:K.matrix.elements} remains valid for independent demands and for cross-price slopes that vanish at the evaluation point.\footnote{The slope ratio $\delta_{jk}=-q_{j,k}/q_{j,j}$ and diversion ratio $D_{j\to k}=-q_{k,j}/q_{j,j}$ need not coincide. They coincide when cross-price derivatives are symmetric, as for differentiable quasilinear demand. No symmetry assumption is used in the decomposition.}

The matrix $\mathbf C$ records additional terms from multiproduct ownership. Its direct-diversion term reflects that changing another owned price changes the margin on recaptured sales. Its markup-weighted diversion derivatives reflect how the price change reallocates lost sales across the portfolio. Both are absent with single-product ownership, but $\mathbf K$ can still create strategic interaction across single-product firms. Thus a change in a portfolio can alter incidence even with demand held fixed, and a change in mixed curvature can alter incidence even with baseline substitution and ownership held fixed. These are the separate economic channels represented by the decomposition.

\paragraph{Preference heterogeneity and the pricing Jacobian.}
For mixed-logit demand, the derivative terms have an interpretation at finite prices. Let $s_{ij}(\mathbf p)$ be type $i$'s probability of buying $j$, with a fixed type distribution $F$ and market size $M$, so that $q_j=M\int s_{ij}\,dF$. Write $\mathbb E_j$ for expectation over types weighted by $s_{ij}/\int s_{ij}\,dF$. With differentiation under the integral justified and the displayed moments finite,
\begin{equation}\label{equation:mixture.composition}
\nabla_{\mathbf p}^2\log q_j
=\mathbb E_j[\nabla_{\mathbf p}^2\log s_{ij}]
+\operatorname{Cov}_j(\nabla_{\mathbf p}\log s_{ij}).
\end{equation}
The first term averages changes in individual choice sensitivities. The covariance term records how repricing changes the composition of buyers. Its diagonal is nonnegative, whereas its off-diagonal entries can have either sign. Both own and cross curvature therefore depend on the consumers purchasing the product, not just its aggregate share.

This is a reformulation of standard mixture differentiation; the own-price calculation corresponds to \citet[equation (14)]{miravete2023pass} and \citet[equation (16)]{miravete2026elasticity}. Here its role is to identify the components entering $\mathbf K$, which equilibrium pricing combines with ownership through $\mathbf C$ and $\Lambda$. Appendix~\ref{section:proof.of.lemma.2} derives the identity and its mixed-logit specializations. The calculation allows jointly distributed product tastes and price sensitivities; it requires neither small shares nor the scalar-mixing restrictions used later for the ray limits.

Together with $\Psi=-\mathbf J_{\mathbf f}^{-1}\Lambda$, the decomposition maps these primitives into equilibrium responses. A pricing direction that produces little change in the first-order conditions can generate a large price response if the cost shock loads on that direction. The next two subsections distinguish within-firm profit curvature from the feedback that couples firms' pricing decisions.

\subsection{Directional Profit Curvature}
The full pricing Jacobian describes equilibrium interactions. To interpret its within-firm blocks, fix rivals' prices and vary firm $f$'s prices along $\mathbf p_f+tv$. Write $\phi_{j,v}(t)=q_j(\mathbf p_f+tv,\mathbf p_{-f})$ and $\Phi_{f,v}(t)=\Pi_f(\mathbf p_f+tv,\mathbf p_{-f})$. Let $m_j=p_j-c_j$ denote an absolute markup. Define
\[
S_f(v):=-v^\top\overline{\mathbf J}_{\mathbf q_f}v,
\qquad M_j(v):=\frac{m_j}{q_j}\phi'_{j,v}(0)^2,
\qquad M_f(v):=\sum_{j\in\mathcal J_f}M_j(v),
\]
where $\overline{\mathbf J}_{\mathbf q_f}=(\mathbf J_{\mathbf q_f}+\mathbf J_{\mathbf q_f}^\top)/2$. The slope term $S_f$ captures the first-order loss of sales along the price direction. The weight $M_j$ equals product $j$'s variable profit times its squared directional semi-elasticity.

For any direction with $M_f(v)>0$, define
\[
R_f(v):=\frac{S_f(v)}{M_f(v)},
\qquad
\kappa_f(v):=\frac{\sum_{j\in\mathcal J_f}m_j\phi''_{j,v}(0)}{M_f(v)}.
\]
The second definition is valid even when a product has zero first-order demand response along $v$. If every $\phi'_{j,v}(0)$ is nonzero, it also has the weighted-curvature representation
\[
\kappa_f(v)=\sum_{j\in\mathcal J_f}\lambda_j(v)\kappa_j(v),
\qquad \lambda_j(v)=\frac{M_j(v)}{M_f(v)},
\qquad \kappa_j(v)=\frac{q_j\phi''_{j,v}(0)}{\phi'_{j,v}(0)^2}.
\]
Nonnegative markups make these weights nonnegative and sum to one. Defining the aggregate directly avoids undefined products of a zero weight and an undefined product-level curvature.

\begin{lemma}[Directional curvature and the second-order condition]
\label{lemma:directional.soc}
Suppose $q_j>0$ and $m_j\geq0$ for all $j\in\mathcal J_f$. For every direction $v$ with $M_f(v)>0$,
\begin{equation}\label{equation:directional.second.derivative.of.profit.function}
\Phi''_{f,v}(0)=M_f(v)\bigl(\kappa_f(v)-2R_f(v)\bigr).
\end{equation}
Consequently the directional second-order condition is equivalent to
\begin{equation}\label{equation:directional.curvature.bound}
\kappa_f(v)\leq2R_f(v).
\end{equation}
If all markups are positive and $\mathbf J_{\mathbf q_f}$ is nonsingular, then $M_f(v)>0$ for every nonzero $v$, and the profit Hessian is negative semidefinite if and only if this inequality holds for every such direction.
\end{lemma}
\begin{proof}
See Appendix \ref{section:proof.of.lemma.1}.
\end{proof}
The ratio $R_f$ compares the first-order sales effect with the scale on which demand curvature enters profits. For one product, the first-order condition implies $M_f(v)=S_f(v)=-q'(p)v^2$, so $R_f(v)=1$ and the inequality reduces to $\kappa\leq2$. With multiple products, the threshold varies with the direction because different joint price changes induce different reallocations of sales. Here ``stability'' refers to an individual firm's local profit curvature, not dynamic stability or invertibility of the whole pricing game. Appendix \ref{app:directional-stability-bounds} gives additional bounds under symmetric demand derivatives.

\begin{remark}[Directional law of demand]
\label{remark:directional_law_of_demand}
If demand satisfies $(\mathbf q(\mathbf p')-\mathbf q(\mathbf p))^\top(\mathbf p'-\mathbf p)\leq0$, its symmetric Jacobian is negative semidefinite. Every principal submatrix then has the same property, giving $S_f(v)\geq0$. To see this, set $\mathbf p'=\mathbf p+tw$, divide the inequality by $t^2$, and let $t\to0$. This condition holds for Hicksian demand at fixed utility and for quasilinear Marshallian demand. Income effects can invalidate it for general Marshallian demand; it is not imposed automatically on every system considered here.
\end{remark}

\subsection{Fixed-Rival Responses and Equilibrium Feedback}
The scalar formula \eqref{equation:pass.though.in.single.product.monopoly} extends to a firm's response holding rival prices fixed. This is a local fixed-rival derivative. Full equilibrium pass-through additionally allows those rivals to adjust.

Fix firm $f$ and a tax perturbation $d\mathbf t_f$. Write $\mathbf J_{\mathbf f_f}=\partial\mathbf f_f/\partial\mathbf p_f^\top$ for the within-firm pricing block and $\Lambda_f$ for the corresponding principal submatrix of $\Lambda$. The fixed-rival response solves
\[
\mathbf J_{\mathbf f_f}d\mathbf p_f+\Lambda_f d\mathbf t_f=0.
\]
Let $\mathbf W_f=-\mathbf J_{\mathbf q_f}^{\mathrm{diag}}\succ0$, define $\|x\|_{\mathbf W_f}=(x^\top\mathbf W_fx)^{1/2}$, and, for a nonzero response, set $\hat v_f=d\mathbf p_f/\|d\mathbf p_f\|_{\mathbf W_f}$.

\begin{proposition}[Directional fixed-rival response]
\label{prop:directional.pass.through}
At an interior stationary point with nonnegative markups and a negative definite own-firm profit Hessian, hold rivals' prices fixed and let $d\mathbf p_f$ denote the derivative along the resulting branch of strict local profit maxima. For $d\mathbf p_f\neq0$ and $M_f(\hat v_f)>0$,
\begin{equation}\label{equation:directional.norm.identity}
\|d\mathbf p_f\|_{\mathbf W_f}
=\frac{\hat v_f^\top\mathbf W_f\Lambda_f d\mathbf t_f}
{M_f(\hat v_f)\bigl(2R_f(\hat v_f)-\kappa_f(\hat v_f)\bigr)}.
\end{equation}
\end{proposition}
\begin{proof}
See Appendix \ref{section:proof.of.proposition.1}.
\end{proof}
When this local branch remains globally optimal, the derivative is also the firm's best-response derivative.
The denominator is minus the own-firm profit curvature along the realized response direction. Holding the direction and projected tax impulse fixed, a smaller denominator produces a larger response. The direction itself is determined by the pricing system, however, so this identity does not replace the matrix calculation with an independently known scalar.\footnote{For one product let $W=-q'>0$. Then $\hat v=\operatorname{sgn}(dp)/\sqrt W$, $R(\hat v)=M(\hat v)=1$, and the numerator is $\operatorname{sgn}(dp)\sqrt W\,dt$. Since $\|dp\|_W=\sqrt W|dp|$, cancelling these factors gives $dp=dt/(2-\kappa)$.}

To recover equilibrium pass-through, let $\mathcal D$ be the block diagonal of $\mathbf J_{\mathbf f}$ containing the full within-firm blocks, and set $\mathcal E=\mathbf J_{\mathbf f}-\mathcal D$. When these blocks are invertible, define
\[
\mathcal T=-\mathcal D^{-1}\mathcal E,
\qquad \Psi^{\mathrm{BR}}=-\mathcal D^{-1}\Lambda.
\]
Factoring $\mathbf J_{\mathbf f}=\mathcal D(I-\mathcal T)$ gives the exact identity
\begin{equation}\label{equation:equilibrium.block.feedback}
\Psi=(I-\mathcal T)^{-1}\Psi^{\mathrm{BR}}.
\end{equation}
Thus own-firm curvature governs the fixed-rival response, while the remaining inverse captures strategic feedback across firms. If $\rho(\mathcal T)<1$, successive powers of $\mathcal T$ represent successive rounds of rival adjustment. This additional condition does not follow from each firm's second-order condition.

A quadratic-utility duopoly gives a globally defined illustration. Appendix~\ref{section:proof.of.proposition.1} specifies consumer demand at all nonnegative prices, including the possibility of zero purchases. When both products are purchased, demand is
\[
q_i=3-2\beta-p_i+\beta p_k,\qquad k\ne i,\quad 0\leq\beta<1.
\]
At marginal costs $(1,1)$, prices $(2,2)$ form a Nash equilibrium: each firm's price is its unique globally optimal response to the rival's price over all $p_i\geq0$. Quantities and markups are one. Throughout this family, own-demand slopes are $-1$, own BP curvature is zero, and own-profit curvature is $-2$ at the equilibrium. Thus fixed-rival own-cost pass-through is always $1/2$. Equilibrium pass-through instead satisfies
\[
\mathbf J_{\mathbf f}=\begin{pmatrix}-2&\beta\\ \beta&-2\end{pmatrix},
\qquad
\Psi=\frac{1}{4-\beta^2}\begin{pmatrix}2&\beta\\ \beta&2\end{pmatrix},
\qquad
\Psi\mathbf1=\frac{1}{2-\beta}\mathbf1.
\]
As substitution strengthens, rival adjustment raises common-cost pass-through even though the individual curvature benchmark is unchanged. At $\beta=3/4$, the three responses are distinct:
\[
\left.\frac{\partial p_i}{\partial c_i}\right|_{p_k\ \mathrm{fixed}}=\frac12,
\qquad
\Psi_{ii}=\frac{32}{55},
\qquad
(\Psi\mathbf1)_i=\frac45.
\]
The best-response slope with respect to the rival's price is $\beta/2=3/8$; this feedback turns the initial one-half response into four-fifths when both costs rise. In fact, for common cost increments $t\in[0,1]$, prices $p_i(t)=2+4t/5$ form a Nash equilibrium. The example therefore illustrates equilibrium amplification over a finite range, with no restriction confining firms to prices near the baseline.

This is an illustration of the established role of strategic feedback in incidence \citep{moorthy2005general,alexandrov2017lechatelier}. It isolates the dependence of equilibrium pass-through on substitution while holding the individual curvature benchmark fixed. For approximation, the relevant question is how much of the joint response is omitted when rival adjustment is suppressed.

\section{Finite-Series Approximation of the Pass-Through Matrix \label{section:finite.series.approximation}}
The economic value of an approximation is that it identifies which interactions are needed for a prediction and checks the error from leaving others out. This section uses the decomposition to organize the equilibrium response into a direct adjustment and successive rounds of interaction. The Neumann expansion is standard matrix algebra; its application here associates each term with demand and ownership effects, supplies error bounds, and establishes convergence throughout the interior common-coefficient logit model. When the full pricing Jacobian is known, a linear solve remains a direct way to compute a particular response. The expansion is useful for interpreting that response and assessing a simpler approximation.

\subsection{Neumann Expansion and Error Control}
Separate the normalized pricing Jacobian into its diagonal and off-diagonal parts:
\begin{equation}\label{equation:Af.Bf.definition}
\mathbf A=-2I+\mathbf K^{\mathrm{diag}}+\mathbf C^{\mathrm{diag}},
\qquad \mathbf B=\mathbf K^{\mathrm{off}}+\mathbf C^{\mathrm{off}}.
\end{equation}
Thus $\mathbf J_{\mathbf f}=\mathbf A+\mathbf B$. The diagonal contains own-price curvature and own-price diversion reallocation. The remaining entries contain cross-price interactions, including those within a multiproduct firm. With $\mathbf A$ invertible, define
\begin{equation}\label{equation:Gamma.structural.definition}
\Gamma=-\mathbf B\mathbf A^{-1}.
\end{equation}
The order matters: $\mathbf J_{\mathbf f}=(I-\Gamma)\mathbf A$, so $\Psi=-\mathbf A^{-1}(I-\Gamma)^{-1}\Lambda$.

\begin{proposition}[Neumann representation and truncation error]
\label{prop:neumann.structural}
Suppose $\mathbf A$ is invertible and $\rho(\Gamma)<1$. Then
\begin{equation}\label{equation:pass.through.neumann.structural}
\Psi=-\mathbf A^{-1}\sum_{m=0}^{\infty}\Gamma^m\Lambda.
\end{equation}
For integer $K\geq0$, define the order-$K$ approximation as
\[
\Psi_K=-\mathbf A^{-1}\sum_{m=0}^{K}\Gamma^m\Lambda.
\]
Its exact remainder is
\begin{equation}\label{equation:neumann.exact.remainder}
\Psi-\Psi_K=-\mathbf A^{-1}\Gamma^{K+1}(I-\Gamma)^{-1}\Lambda.
\end{equation}
If $\|\Gamma\|\leq\eta<1$ in a consistent submultiplicative matrix norm, then
\begin{equation}\label{equation:neumann.error.bound}
\|\Psi-\Psi_K\|\leq\|\mathbf A^{-1}\|\,\|\Lambda\|\frac{\eta^{K+1}}{1-\eta}.
\end{equation}
For a particular cost direction $h$, the same bound replaces $\|\Lambda\|$ with $\|\Lambda h\|$ and the left side with $\|(\Psi-\Psi_K)h\|$, using the compatible vector norm.
\end{proposition}
\begin{proof}
See Appendix \ref{section:proof.of.proposition.2}.
\end{proof}
The first two approximations are
\begin{align}
\Psi_0&=-\mathbf A^{-1}\Lambda,\label{equation:pass.through.zeroth.order.structural}\\
\Psi_1&=-\mathbf A^{-1}\Lambda+\mathbf A^{-1}\mathbf B\mathbf A^{-1}\Lambda.
\label{equation:pass.through.first.order.structural}
\end{align}
The leading term removes propagation through off-diagonal pricing derivatives, while retaining ownership in the initial tax loading $\Lambda$. The next term adds one round of interaction. Equivalently, set $T=-\mathbf A^{-1}\mathbf B=\mathbf A^{-1}\Gamma\mathbf A$. Then $\Psi_K=\sum_{m=0}^KT^m\Psi_0$, so each power acts directly on a vector of price responses. These product-level rounds differ from the firm-level rounds in \eqref{equation:equilibrium.block.feedback}, whose baseline retains entire within-firm blocks.

The convergence restriction is additional to equilibrium regularity and the firms' second-order conditions. Spectral radius below one ensures eventual convergence, but a small spectral radius alone need not make low-order terms small for a nonnormal matrix. The norm bound provides an operational sufficient condition and retains the size of the diagonal inverse and tax loading. It also clarifies the meaning of ``first order'': $\Psi_1$ retains one round of feedback, whereas every $\Psi_Kh$ is a \emph{local} prediction in the size of the policy change. The bound does not control extrapolation to a large finite policy change. Nor must adding a term reduce each entry's error in every demand system.

\begin{example}[An interior logit benchmark]
\label{example:logit.neumann.convergence}
Consider multinomial logit with common price coefficient $\alpha>0$, shares $s_j>0$, and positive outside share. Let $S_f=\sum_{j\in\mathcal J_f}s_j$ and $S=\sum_js_j<1$. The first-order conditions imply a common within-firm markup $m_f=1/[\alpha(1-S_f)]$. In fact,
\[
f_j=\frac{1/\alpha-m_j+\sum_{l\in\mathcal J_f}s_lm_l}{1-s_j},\qquad j\in\mathcal J_f.
\]
At equilibrium the numerator vanishes. Differentiating it while keeping costs fixed gives
\[
A_{jj}=-\frac1{1-s_j},\qquad
B_{jk}=\begin{cases}
0,&k\ne j,\ k\in\mathcal J_f,\\[2pt]
\dfrac{S_fs_k}{(1-s_j)(1-S_f)},&k\notin\mathcal J_f.
\end{cases}
\]
The within-firm off-diagonal pricing derivatives cancel at equilibrium. This cancellation shows why an ownership-aware baseline can be accurate even before adding rival feedback. Ownership still affects both $S_f$ and $\Lambda$. In particular $(\Psi_0)_{jj}=1-s_j$, $(\Psi_0)_{jk}=-s_k$ for another product in the same firm, and other entries vanish.

The similar propagation matrix $T=-\mathbf A^{-1}\mathbf B$ has nonnegative cross-firm entries $T_{jk}=S_fs_k/(1-S_f)$ and zero within-firm entries. Consequently
\begin{equation}\label{equation:logit.all.interior.convergence}
\rho(\Gamma)=\rho(T)\leq\|T\|_\infty
=\max_f\frac{S_f(S-S_f)}{1-S_f}\leq\max_fS_f<1.
\end{equation}
Thus the expansion converges at every such interior logit equilibrium, including multiproduct ownership. This bound concerns $T$, not necessarily $\|\Gamma\|_\infty$. If $\eta=\|T\|_\infty$, the equivalent representation gives the shock-specific certificate
\[
\|(\Psi-\Psi_K)h\|_\infty
\leq\frac{\eta^{K+1}}{1-\eta}\|\Psi_0h\|_\infty.
\]

For three single-product firms with shares $(0.20,0.15,0.10)$,
\[
\Psi\approx\begin{pmatrix}
0.801530&0.032300&0.023147\\
0.028612&0.851403&0.016713\\
0.018289&0.014908&0.900793
\end{pmatrix},\qquad
\Psi_1\approx\begin{pmatrix}
0.800000&0.031875&0.022500\\
0.028235&0.850000&0.015882\\
0.017778&0.014167&0.900000
\end{pmatrix}.
\]
The order-1 correction captures most cross-product responses in this example. More broadly, Appendix \ref{section:numerical.illustration} reports 1,000 validated six-product logit configurations: mean relative matrix error for order 1 is $0.323\%$. Logit's exact within-firm cancellation makes this a favorable benchmark, not evidence of universal accuracy in flexible demand systems. The appendix reports finite-policy checks separately. Exact aggregation for logit and CES is an established route to tractability \citep{nocke2018multiproduct,nocke2025aggregative}; the calculation here uses logit to make the convergence restriction and the contributions of successive interactions explicit.
\end{example}

\begin{remark}[Block-diagonal refinement]
When within-firm interactions are substantial, the block decomposition in \eqref{equation:equilibrium.block.feedback} may give a more informative baseline. If $\rho(\mathcal T)<1$, its Neumann expansion is $\Psi=\sum_{m\geq0}\mathcal T^m\Psi^{\mathrm{BR}}$. The same remainder argument applies, with firm blocks replacing diagonal entries.
\end{remark}

\subsection{Checking a Particular Response}

A researcher may need the response of one price, a price index, or the price component of consumer welfare rather than every element of $\Psi$. The same contraction argument gives a check on the particular calculation. Write the linearized pricing system as
\[
\mathbf J_{\mathbf f}x+\mathbf b=0,
\]
where $x$ is the exact local price response and $\mathbf b$ is the intervention's change in pricing incentives at fixed prices. For a cost direction $h$, $\mathbf b=\Lambda h$; Section~\ref{subsection:general.incentives} gives the loading for other interventions. For any proposed approximation $\widehat x$, define its residual $r=\mathbf J_{\mathbf f}\widehat x+\mathbf b$.

Suppose $T=-\mathbf A^{-1}\mathbf B$ satisfies $\|T\|\leq\eta<1$ in the matrix norm induced by the chosen vector norm. Since $\mathbf J_{\mathbf f}=\mathbf A(I-T)$,
\[
x-\widehat x=-(I-T)^{-1}\mathbf A^{-1}r,
\qquad
\|x-\widehat x\|\leq\frac{\|\mathbf A^{-1}r\|}{1-\eta}.
\]
For a fixed vector of weights $\ell$, the dual norm gives
\begin{equation}\label{equation:policy.residual.certificate}
\bigl|\ell^\top(x-\widehat x)\bigr|
\leq \|\ell\|_*\frac{\|\mathbf A^{-1}r\|}{1-\eta}
=:E_\ell.
\end{equation}
Consequently, $|\ell^\top\widehat x|>E_\ell$ establishes the sign of the local scalar response. This is a computational consequence of the contraction condition: an approximation can be adequate for a particular economic conclusion even when a bound on the entire matrix is uninformative. Choosing $\ell=e_j$ checks a product's price response; suitable fixed weights check a price index. Choosing $\ell=-\mathbf q$ checks the price-mediated consumer-welfare component when that interpretation applies. The check conditions on the pricing derivatives and intervention loading. Estimation uncertainty, direct utility effects, and extrapolation to a finite policy change require separate treatment. Applied to estimated demand, this bound can establish that interactions omitted by a simpler calculation do not overturn a particular price or index prediction. It provides a conditional check on the importance of feedback rather than inferring that importance from own-curvature estimates alone.

\section{Small-Share Limits of Pass-Through \label{section:small.share.approximation}}

Small product shares do not by themselves imply weak strategic interaction. The relevant question is how shares become small and which substitution patterns survive along that path. This section compares demand specifications through the resulting limits of the normalized pricing Jacobian. In some fixed-product experiments, pass-through becomes diagonal; in others, cross-product effects and ownership remain relevant even when every inside share vanishes.

The analysis distinguishes demand derivatives evaluated along a price sequence from equilibrium comparative statics. All matrix limits below keep the number of products fixed. To interpret a demand-side limit as equilibrium pass-through, the sequence must consist of interior equilibria, possibly supported by changing costs, and the limiting full pricing Jacobian must be nonsingular. A demand-evaluation path alone does not establish those equilibrium conditions.

These comparisons connect the local decomposition to the demand-shape concerns in \citet{miravete2023pass,miravete2026elasticity,compiani2022market}. The issue is which restrictions a specification places on the matrix used for a counterfactual, even when its baseline elasticities appear reasonable. I first isolate the demand-side terms, distinguish baseline curvature from an entire normalized profile, and then incorporate ownership. The final nested-logit calculation shows how to retain the remaining interactions in a simple formula.

\subsection{A Semi-Elasticity Representation of Pass-Through with Single-Product Firms}

When each firm produces a single product, \(\Lambda=I\) and \(\mathbf C=\mathbf 0\), so Lemma~\ref{lemma:matrix.curvature.decomposition} implies that the pass-through matrix simplifies to
\[
\frac{\partial \mathbf p}{\partial \mathbf t^\top} = - \mathbf{J}_{\mathbf{f}}^{-1} = -\left[ -2I + \mathbf{K} \right]^{-1},
\]
so the entries of $\mathbf{J}_{\mathbf{f}} \equiv \partial \mathbf{f} / \partial \mathbf{p}^\top$ are
\begin{equation}\label{equation:normalized.pricing.jacobian.single.product}
J_{\mathbf f,jk}
=
\begin{cases}
-2+\kappa_{jj}^j, & k=j,\\
\delta_{jk}\bigl(1-\kappa_{kj}^j\bigr), & k\neq j.
\end{cases}
\end{equation}
Thus, with single-product firms, the behavior of pass-through is governed by the diagonal and off-diagonal entries of the normalized pricing Jacobian in \eqref{equation:normalized.pricing.jacobian.single.product}.

A convenient way to study these objects under small-share asymptotics is through semi-elasticities. For each product \(j\), define
\[
\eta_{jj}(\mathbf p):=\partial_{j}\log q_j(\mathbf p)
=
\frac{\partial_{p_j} q_j(\mathbf p)}{q_j(\mathbf p)},
\qquad
\eta_{jk}(\mathbf p):=\partial_{k}\log q_j(\mathbf p)
=
\frac{\partial_{p_k} q_j(\mathbf p)}{q_j(\mathbf p)},
\quad k\neq j.
\]
Thus \(\eta_{jj}\) is the own semi-elasticity of demand for product \(j\), which is negative under downward-sloping demand, and \(\eta_{jk}\) is the cross semi-elasticity of demand for product \(j\) with respect to \(p_k\).

\begin{lemma}[Semi-elasticity representation of the normalized pricing Jacobian] \label{lemma:semi.elasticity.representation}
Assume demand is positive and twice continuously differentiable, with a negative own-price derivative. For each product \(j\),
\[
-2+\kappa_{jj}^j(\mathbf p)
=
-1+\frac{\partial_{p_j}\eta_{jj}(\mathbf p)}{\eta_{jj}(\mathbf p)^2}.
\]
For each \(k\neq j\),
\[
\delta_{jk}(\mathbf p)\bigl(1-\kappa_{kj}^j(\mathbf p)\bigr)
=
\frac{\partial_{p_j}\eta_{jk}(\mathbf p)}{\eta_{jj}(\mathbf p)^2}.
\]
When a cross derivative is zero, the off-diagonal expression is understood through the continuous combined formula
\[
\delta_{jk}(1-\kappa_{kj}^j)
=\frac{q_j\partial_{jk}q_j}{(\partial_jq_j)^2}
-\frac{\partial_kq_j}{\partial_jq_j},
\]
which does not require division by a cross slope.
\end{lemma}

\begin{proof}
See Appendix \ref{section:proof.of.lemma.3}.
\end{proof}

Lemma \ref{lemma:semi.elasticity.representation} shows that the normalized pricing Jacobian is governed by two local growth-rate objects. The diagonal term depends on the normalized derivative of the own semi-elasticity, \(\partial_{p_j}\eta_{jj}/\eta_{jj}^2\), while the off-diagonal term depends on the normalized derivative of the cross semi-elasticity, \(\partial_{p_j}\eta_{jk}/\eta_{jj}^2\). This representation is helpful because it rewrites the Jacobian in terms of the local evolution of economically interpretable demand sensitivities, rather than raw second derivatives of demand. In particular, for small-share analysis, the key question is not the absolute size of these derivatives, but how own and cross semi-elasticities change relative to the square of the own semi-elasticity along the specified sequences of price vectors.

\subsection{Baseline Curvature and Normalized Demand Profiles}

Consider a sequence of price vectors \(\mathbf p^{(n)}\) along which inside shares vanish. Suppose the normalized coefficients converge:
\[
-2+\kappa_{jj}^j(\mathbf p^{(n)})\to-1+a_j,
\qquad
\delta_{jk}(\mathbf p^{(n)})(1-\kappa_{kj}^j(\mathbf p^{(n)}))\to b_{jk}.
\]
These are assumptions on a demand system and a particular sequence; smoothness and small shares alone do not guarantee their existence. The numbers \(a_j\) and \(b_{jk}\) describe derivatives at the baseline price vectors. An entire own-price demand profile requires information about derivatives away from those baselines.

To make the distinction precise, let \(\lambda_{j,n}:=-\eta_{jj}(\mathbf p^{(n)})>0\) and vary only product \(j\)'s price by \(t/\lambda_{j,n}\). This measures the price change in units of the inverse baseline semi-elasticity, so the normalized log-demand slope at zero is always \(-1\). Define
\[
\mathbf p_{j,n}(t):=\mathbf p^{(n)}+\frac{t}{\lambda_{j,n}}e_j,
\qquad
E_{j,n}(t):=\frac{\eta_{jj}(\mathbf p_{j,n}(t))}{\eta_{jj}(\mathbf p^{(n)})},
\qquad
Q_{j,n}(t):=\frac{q_j(\mathbf p_{j,n}(t))}{q_j(\mathbf p^{(n)})},
\]
Under the uniform-neighborhood conditions stated in Proposition~\ref{prop:local.shape.jacobian}, normalized curvature is asymptotically constant on an interval $I$ containing zero. The resulting profile is
\[
Q_{j,n}(t)\longrightarrow
\begin{cases}
e^{-t},&a_j=0,\\
(1+a_jt)^{-1/a_j},&a_j\ne0,
\end{cases}
\]
uniformly on compact subintervals of $I\cap\{t:1+a_jt>0\}$. Appendix~\ref{section:proof.of.proposition.3} gives the full conditions, the associated semi-elasticity limits, and the proof.

The uniform neighborhood assumptions do the substantive work. Under them, \(a_j=0\) gives an exponential profile, \(a_j>0\) gives a power profile, and \(a_j<0\) gives a profile tending to zero as \(t\uparrow-1/a_j\), if the common domain extends to that boundary. Compact convergence before the boundary does not imply a literal choke price in each original demand curve. Nor does a baseline value of \(a_j\) establish any of these profiles without the neighborhood assumptions.

More generally, write \(A_{j,n}(t)=-E'_{j,n}(t)/E_{j,n}(t)^2\). The exact identities
\[
\frac{1}{E_{j,n}(t)}=1+\int_0^t A_{j,n}(u)\,du,
\qquad
Q_{j,n}(t)=\exp\!\left[-\int_0^t E_{j,n}(u)\,du\right]
\]
show why the entire function \(A_{j,n}\), rather than only its value at zero, matters for demand away from the baseline. At the baseline, the precise local implication is
\[
Q_{j,n}(t)=1-t+\frac{1}{2}\kappa_{jj}^j(\mathbf p^{(n)})t^2+o_n(t^2).
\]
Thus \(a_j\) compares baseline curvature with the exponential benchmark. This distinction matters for counterfactuals: a local derivative concerns behavior at the baseline, whereas extrapolation along a price path requires demand shape away from it. Gamma mixed logit and nested logit below show why a baseline curvature index cannot substitute for that additional information.

\subsection{Common Demand Systems and Small-Share Paths}

Table \ref{table:small.share.limits} compares the limiting demand-side entries for explicit experiments. Logit, CES, and mixed logit are evaluated along positive common-price rays \(\mathbf p^{(n)}=\lambda_n\mathbf p\), with fixed mean utilities and \(\lambda_n\to\infty\). Linear and fixed-weight Stone LA/AIDS demand instead approach a feasible boundary for one product. Nested logit lets all nest shares vanish while conditional shares within nests converge. Appendix \ref{appendix:model.specific.small.share.limits} states the demand systems and proves the limits. In the mixed-logit specifications, mixing is over a scalar price coefficient, with no random coefficients on other product attributes. Those limits concern the specified mixtures and do not cover unrestricted correlated random-coefficient demand.

\begin{table}[htbp!]
\centering
\small
\caption{Demand-Side Limits Under Specified Small-Share Experiments}
\label{table:small.share.limits}
\begin{threeparttable}
\begin{tabular}{lcc}
\toprule
Model and experiment & \(\lim(-2+\kappa_{jj}^j)\) & \(\lim\delta_{jk}(1-\kappa_{kj}^j)\), \(k\ne j\)\\
\midrule
\multicolumn{3}{@{}l}{\emph{Positive common-price ray; fixed finite product set}}\\[0.3em]
Logit & \(-1\) & \(0\)\\[0.5em]
Lognormal mixed logit & \(-1\) & \(0\)\\[0.5em]
CES, \(\sigma>1\) & \(-1+1/\sigma\) & \(0\)\\[0.5em]
Gamma mixed logit & \(-1+\partial_{jj}\log H_j/\ell_j^2\) & \(\partial_{jk}\log H_j/\ell_j^2\)\\
\midrule
\multicolumn{3}{@{}l}{\emph{Feasible boundary for product \(j\)}}\\[0.3em]
Linear demand & \(-2\) & \(-\beta_{jk}/\beta_{jj}\)\\[0.5em]
Stone LA/AIDS & \(-2\) & \(-\overline c_{jk} A_{jk}/A_{jj}\)\\
\midrule
\multicolumn{3}{@{}l}{\emph{Vanishing nest shares; stable conditional shares}}\\[0.3em]
Nested logit & \(-1-\sigma z_j(1-z_j)/B_j^2\) & \(\sigma z_jz_k/B_j^2\) within nest\\
 & & \(0\) across nests\\
\bottomrule
\end{tabular}
\begin{tablenotes}[flushleft]
\footnotesize
\item \emph{Notes:} All entries concern the demand-side matrix \(-2I+\mathbf K\). For gamma mixing, \(H_j\) is the integral in Proposition \ref{prop:ray.limits.mixed.logit}, evaluated at the fixed ray \(\mathbf p\), and \(\ell_j=-\partial_j\log H_j>0\). Its baseline curvature correction can have either sign. For Stone LA/AIDS, weights are fixed, \(A_{jk}=\gamma_{jk}-\beta_j\omega_k\), \(A_{jj}<0\), and \(\overline c_{jk}=\lim p_j^{(n)}/p_k^{(n)}\in(0,\infty)\). Those boundary statements are rowwise. For nested logit, \(z_j\) is the limiting conditional share and \(B_j=1-\sigma z_j\). Turning these entries into equilibrium pass-through additionally requires the full ownership terms and a nonsingular limiting pricing Jacobian.
\end{tablenotes}
\end{threeparttable}
\end{table}

Logit and lognormal mixed logit share the demand-side limit \(-I\); CES has limit \(-(1-1/\sigma)I\). With single-product firms these yield unit pass-through limits of \(I\), \(I\), and \(\sigma/(\sigma-1)I\), respectively, along regular equilibrium sequences. Their similarity in matrix structure does not mean that they impose identical finite-share curvature or approximation errors. In mixed logit, the buyer-composition terms in \eqref{equation:mixture.composition} can remain substantial at finite prices. The lognormal high-price-ray result therefore does not require the finite-market pass-through predictions studied by \citet{miravete2023pass,miravete2026elasticity} to be near the identity. Its restrictions concern a particular selection experiment, rather than small observed shares alone.

Gamma mixed logit gives a different mechanism. Along the common ray, relevant consumers have price sensitivity of order \(1/\lambda_n\). Rival inside products can therefore retain substantial conditional choice probabilities for the consumers who remain, even though aggregate shares vanish. Demand is regularly varying along that ray,
\[
q_j(\lambda\mathbf p)\sim\frac{M\beta^r}{\Gamma(r)}\lambda^{-r}H_j(\mathbf p),
\]
but own-price perturbations change relative prices. The resulting normalized residual-demand profile is
\[
Q_j(t)=\frac{H_j(\mathbf p+t e_j/\ell_j)}{H_j(\mathbf p)},
\qquad t>-\ell_jp_j,
\]
which generally has nonconstant normalized curvature. In particular, \(a_j=\partial_{jj}\log H_j/\ell_j^2\) need not be positive. With two equally attractive products, \(r=2\), \(\delta_1=\delta_2=10\), and ray prices \((1,1)\), numerical quadrature gives \(a_j\simeq-0.188\). A common-ray power law therefore does not determine whether residual demand is more or less curved than the exponential benchmark. This sensitivity to the mixing distribution complements the demand-curvature comparisons in \citet{miravete2026elasticity}.

Boundary experiments preserve a different source of interaction. Linear and Stone LA/AIDS demand retain cross-slope ratios in the normalized Jacobian as product \(j\)'s demand or expenditure share approaches zero. These are statements about a feasible rowwise boundary; they do not automatically describe a simultaneously feasible whole-market equilibrium limit. A finite choke-price interpretation further requires finite positive limiting prices.

Finally, nested logit preserves substitution within nests. Its negative baseline curvature correction does not imply a finite-choke normalized profile. Writing \(z=z_j\) and \(B=1-\sigma z\), the actual own-price profile is
\[
Q_j(t)=e^{-t/B}\bigl(1-z+ze^{-t/B}\bigr)^{-\sigma},
\]
which is positive at every finite \(t\). Appendix \ref{appendix:nested.logit.benchmark} derives this profile and the corresponding ownership limit. The common lesson is that small-share pricing depends on the limiting experiment and the substitution patterns it preserves, not on a scalar tail label alone.

\subsection{Extension to Multiproduct Firms}

With multiproduct ownership,
\[
\mathbf P=-\mathbf J_{\mathbf f}^{-1}\Lambda,
\qquad
\mathbf J_{\mathbf f}=-2I+\mathbf K+\mathbf C.
\]
The demand-side limit of \(-2I+\mathbf K\) is only one part of this system. A diagonal demand-side limit does not imply that diversion or its markup-weighted derivatives vanish.

A sufficient set of conditions is transparent. At fixed dimension, suppose the demand-side diagonal entries converge to \(-d_j\), with \(d_j>0\), the off-diagonal demand-side entries vanish, and, for all relevant indices \(\ell\ne j\) and \(k\),
\[
D_{j\to\ell}\to0,
\qquad
(p_\ell-c_\ell)\partial_{p_k}D_{j\to\ell}\to0.
\]
Then \(\mathbf C\to0\), \(\Lambda\to I\), and continuity of matrix inversion gives
\[
\mathbf P\to\operatorname{diag}(d_j^{-1}).
\]
Appendix \ref{appendix:model.specific.small.share.limits} verifies these ownership conditions for logit, CES, and lognormal mixed logit along the stated rays, provided markups grow at most proportionally to the price scale. This is a fixed-product experiment in which total inside demand vanishes; it does not cover a growing product set with small individual shares but substantial total firm shares.

Nested logit provides a useful case in which within-nest diversion survives. Let \(S_{fg}\) be firm \(f\)'s limiting share within nest \(g\), let \(\alpha>0\) be the common price coefficient, and let \(\sigma\in[0,1)\) be the nesting parameter. Proposition \ref{prop:nested.logit.small.share} shows that products in the same firm--nest group have the common limiting markup
\[
m_{fg}^*=\frac{1-\sigma}{\alpha(1-\sigma S_{fg})}.
\]
Ownership changes markups through the share of within-nest substitution internalized by the firm.\footnote{For a firm confined to one nest, this is the vanishing-nest limit of the markup formula in \citet[Appendix B, equation (B12)]{nocke2025aggregative}, after restoring price units. For arbitrary firm--nest overlaps, the same limiting markup follows from the reduced pricing equations in \citet{garrido2024aggregative}. Proposition~\ref{prop:nested.logit.small.share} derives the limiting price-response matrix; the following corollary expresses its action through group averages.} The limiting pass-through matrix is block diagonal by nest, and every block has row sums equal to one: a common cost increase within a nest is passed through fully. If one firm owns the entire nest, its limiting markup is the constant \(1/\alpha\), and the block is the identity. Under split ownership, product-specific shocks can still produce within-nest spillovers. Thus surviving substitution creates scope for ownership to matter, while the cost-shock direction and ownership pattern determine its actual incidence.

\subsection{A Group-Average Formula for Nested-Logit Incidence}

The surviving within-nest interactions have a simple economic representation. Fix a nest $g$ in the limit of Proposition~\ref{prop:nested.logit.small.share}. Let $z_j$ be its limiting conditional product shares, with $\sum_{j\in\mathcal J_g}z_j=1$, and write $S_f=\sum_{j\in\mathcal J_f\cap\mathcal J_g}z_j$ for the firm's share of this nest. The following implication reduces the limiting response to averages within ownership groups.

\begin{corollary}[Nested-logit group responses]\label{corollary:nested.group.response}
Under the conditions of Proposition~\ref{prop:nested.logit.small.share}, fix a cost direction $h$ within nest $g$. Define
\[
\chi_f=\left[1+\frac{\sigma S_f}{(1-\sigma S_f)^2}\right]^{-1},
\qquad
\overline h_f=\frac{\sum_{k\in\mathcal J_f\cap\mathcal J_g}z_kh_k}{S_f}
\quad\text{for }S_f>0,
\]
and
\[
D_g=\sum_{f:S_f>0}S_f\chi_f,
\qquad
H_g=\sum_{f:S_f>0}\frac{S_f\chi_f}{D_g}\,\overline h_f.
\]
The limiting local response is
\begin{equation}\label{equation:nested.group.response}
(P_g^*h)_j=h_j+(1-\chi_f)(H_g-\overline h_f),
\qquad j\in\mathcal J_f\cap\mathcal J_g,\quad S_f>0.
\end{equation}
For a group with $S_f=0$, the response is simply $h_j$. Within the nest, cross-firm pass-through entries are nonnegative, within-firm off-diagonal entries are nonpositive, and $0<(P_g^*)_{jj}\leq1$.
\end{corollary}
\begin{proof}
See Appendix~\ref{appendix:nested.group.response.proof}.
\end{proof}

The response has two parts: the product's own cost change and a common markup adjustment within its firm--nest group. That adjustment compares the group's average cost change with a weighted average across the nest. If only one product's cost rises, rival groups raise prices weakly, while its owner weakly lowers the prices of other products in the same group. For a common cost change, the two averages coincide and every price moves one for one. If a firm owns the entire nest, the averages coincide for every cost direction and $P_g^*=I$. A product with $z_j=0$ can still share its positive-share group's markup adjustment; an individual zero share is different from a zero group share.

Related spillover signs already appear in \citet[equation (6)]{moorthy2005general}. In his displayed nested-logit solution, each of two retailers owns all the brands in its nest, and there is no outside good. A cost increase on one brand at one retailer generates a negative within-retailer cross-price response and positive cross-retailer responses; a common cost change passes through fully. Here, nest shares vanish relative to an outside option, and ownership can be split within nests or span several nests. The identity block under complete ownership of a vanishing nest reflects the disappearance of between-nest diversion along this path. It does not contradict the cross-retailer responses in his market with positive retailer shares.

Computing a limiting response requires group averages rather than a product-level matrix inversion, while preserving the sign and size of the remaining portfolio spillovers. The formula is an implication of the specified nested benchmark, not an approximation guarantee at positive nest shares. Appendix~\ref{section:numerical.illustration} evaluates that separate approximation question.

\section{Applications \label{section:applications}}

A response matrix is useful across economic questions because different interventions can change the same pricing conditions. This section first specifies that common calculation, then applies it to demand shifts, incidence with revenue and margin data, and ownership changes. The distinction between an intervention's initial incentive and the matrix that transmits it follows the general incidence approach of \citet{weyl2013pass,adachi2022pass} and the local merger approach of \citet{jaffe2013first}. The contribution here is to evaluate and interpret that transmission through the demand and ownership structure developed above.

\subsection{General Changes in Pricing Incentives}\label{subsection:general.incentives}

Let a scalar parameter $\theta$ smoothly change demand, product-specific marginal costs, or the weights on internalized profits. Keep the product set fixed and suppose the pricing equations admit a differentiable equilibrium branch at the baseline, with $\mathbf J_{\mathbf f}$ nonsingular. The own-price demand slopes remain negative locally, and demand is twice continuously differentiable in prices and $\theta$. Write
\[
F_j(\mathbf p,\theta)=q_j(\mathbf p,\theta)
 +\sum_l\Omega_{jl}(\theta)(p_l-c_l(\theta))q_{l,j}(\mathbf p,\theta),
\qquad
f_j=-F_j/q_{j,j},
\]
where $\Omega_{jj}=1$. At equilibrium $F_j=0$, so differentiation at fixed prices gives the incentive loading
\begin{equation}\label{equation:general.policy.loading}
b_j:=\partial_\theta f_j
=-\frac{q_{j,\theta}
 +\sum_l\Omega_{jl}m_lq_{l,j\theta}
 +\sum_l\Omega_{jl,\theta}m_lq_{l,j}
 -\sum_l\Omega_{jl}c_{l,\theta}q_{l,j}}{q_{j,j}}.
\end{equation}
All terms are evaluated at the baseline. The derivative of the normalization drops out because it multiplies $F_j=0$. Implicit differentiation therefore yields
\begin{equation}\label{equation:general.policy.response}
\frac{d\mathbf p}{d\theta}=-\mathbf J_{\mathbf f}^{-1}\mathbf b.
\end{equation}
For a pure cost change, $\mathbf b=\Lambda\mathbf c_\theta$. More generally, if $\Lambda$ is invertible, the vector $h=\Lambda^{-1}\mathbf b$ is a cost-equivalent incentive and the response is $\Psi h$. Equation~\eqref{equation:general.policy.response} remains valid without this extra invertibility condition. The Neumann approximation in Section~\ref{section:finite.series.approximation} applies directly with $\mathbf b$ replacing $\Lambda h$.

The demand terms in \eqref{equation:general.policy.loading} identify what a demand-changing intervention requires. Its effect on sales, $q_{j,\theta}$, is only one component; the effect on price sensitivities, $q_{l,j\theta}$, also changes the marginal value of repricing. This distinction matters for smooth changes in product quality, consumer exposure, or other demand shifters. An experiment that measures sales effects at unchanged prices need not reveal the resulting pricing response. The formula specifies the additional derivatives required, conditional on the maintained pricing model.

\paragraph{A demand-translation illustration.}
Suppose the intervention has the particular form $\mathbf q(\mathbf p,\theta)=\mathbf Q(\mathbf p-v\theta)$, with fixed costs and ownership and a fixed vector $v$ in price units. Setting $\widetilde{\mathbf p}=\mathbf p-v\theta$ changes effective marginal costs to $\widetilde{\mathbf c}=\mathbf c-v\theta$, while preserving every markup and demand function in the transformed pricing problem. Hence
\begin{equation}\label{equation:demand.translation.response}
\frac{d\mathbf p}{d\theta}=(I-\Psi)v.
\end{equation}
For a shift in only product $k$'s willingness to pay, its own price response is $(1-\Psi_{kk})v_k$, and another product's response is $-\Psi_{jk}v_k$. The identity shows how the same matrix governs a demand-side intervention with a different initial incentive. Additive mean-utility changes $\beta_j\theta$ in logit with a common price coefficient $\alpha$ provide an example, with $v_j=\beta_j/\alpha$. Arbitrary demand shifts, including changes in price sensitivities, instead require the general loading above.

For interventions that change utility directly, $-\mathbf q^\top d\mathbf p$ measures only the price-mediated consumer-welfare effect. The direct utility effect must also enter a total welfare calculation. The next two applications hold demand primitives fixed and state the welfare interpretation needed there.

\subsection{Consumer Welfare with Revenue and Margin Data}

Revenues and relative margins can be more readily available than separate product prices and quantities. This motivates expressing incidence in percentage units, as in the related merger application of \citet{koh2024merger}. The transformation itself is exact. Using it with limited data, however, requires restrictions on pass-through: revenues and margins alone do not identify demand curvature or the equilibrium response matrix. Here logit and CES provide transparent examples of what those restrictions imply.

Assume positive baseline prices and marginal costs, and define proportional cost shifters by
\[
c_j(\tau_j)=c_j e^{\tau_j},
\qquad
\Psi^\tau\equiv\frac{\partial\log\mathbf p}{\partial\boldsymbol\tau^\top}.
\]
A differential $d\tau_j$ is a proportional marginal-cost change. Let $\Psi$ denote unit-cost pass-through, and write $\widehat{\mathbf p}=\operatorname{diag}(p_j)$ and $\widehat{\mathbf c}=\operatorname{diag}(c_j)$.

\begin{lemma}[Percentage pass-through and unit-cost pass-through]
\label{lem:percentage_pt}
At a regular interior equilibrium with positive prices and costs, and at $\boldsymbol\tau=\mathbf 0$,
\[
\Psi^\tau=\widehat{\mathbf p}^{-1}\Psi\widehat{\mathbf c}.
\]
\end{lemma}
\begin{proof}
See Appendix \ref{section:proof.of.lemma.4}.
\end{proof}

Each column of this expression converts a proportional cost change into a unit-cost change; each row converts the resulting price change into a proportional price change. It is a change of coordinates, rather than an identification result.

To interpret model comparisons such as the finite cost experiment in \citet[Section~4.2]{miravete2023pass}, both the shock and the reported units must be matched. A common proportional cost change has local price response $\Psi\mathbf c$ and percentage response $\Psi^\tau\mathbf1=\widehat{\mathbf p}^{-1}\Psi\mathbf c$. A finite cost change additionally requires an equilibrium calculation or control of the linearization error. The following model-specific limits clarify the relation between unit and percentage pass-through.

\begin{lemma}[Unit-cost pass-through under logit and CES]
\label{lem:unit_pt_logit_ces}
Fix a finite number of products, an ownership structure, and demand parameters. Consider sequences of regular interior equilibria with positive prices and costs and nonnegative markups.
\begin{enumerate}
\item[(i)] Under multinomial logit with a common price coefficient, if quantity shares satisfy $\|\mathbf s\|\to0$, then
\[
\Psi_{\mathrm{logit}}=I+O(\|\mathbf s\|).
\]
\item[(ii)] Under finite-market CES with substitution parameter $\sigma>1$, if expenditure shares satisfy $\|\mathbf w\|\to0$ and all relative prices $p_j/p_k$ remain uniformly bounded, then
\[
\Psi_{\mathrm{CES}}=\frac{\sigma}{\sigma-1}I+O(\|\mathbf w\|).
\]
\end{enumerate}
\end{lemma}
\begin{proof}
See Appendix \ref{section:proof.of.lemma.5}.
\end{proof}

These are sequences of equilibria, potentially supported by changing costs. Evaluating demand at arbitrary high prices with fixed costs does not by itself produce such a sequence. Fixed product dimension also matters: small individual shares need not eliminate portfolio effects when the number of products grows.

\begin{proposition}[Percentage pass-through under logit and CES]
\label{prop:percentage_pt_logit_ces}
Under the corresponding hypotheses of Lemma \ref{lem:unit_pt_logit_ces}, assume in addition that relative prices $p_j/p_k$ and cost shares $c_j/p_j$ remain uniformly bounded. Then
\begin{align*}
\Psi^\tau_{\mathrm{logit}}
&=\operatorname{diag}(1-\mu_j)+O(\|\mathbf s\|),\\
\Psi^\tau_{\mathrm{CES}}
&=I+O(\|\mathbf w\|),
\end{align*}
where $\mu_j=(p_j-c_j)/p_j$ is the relative margin.
\end{proposition}
\begin{proof}
See Appendix \ref{section:proof.of.proposition.4}.
\end{proof}

The comparison concerns limiting responses. In logit, a unit-cost increase is approximately passed through one-for-one, so percentage pass-through is the cost share $1-\mu_j$. Under CES, unit-cost pass-through approaches $\sigma/(\sigma-1)$ while the cost share approaches $(\sigma-1)/\sigma$; these factors cancel. Bounded relative prices and cost shares control the rescaling error without requiring absolute prices to stay bounded.

\paragraph{A finite-share CES benchmark.}
Exact percentage pass-through of one holds under the separate assumption of exogenous isoelastic residual demand $q=Ap^{-\sigma}$ with $A$ fixed, since optimal price is then $\sigma c/(\sigma-1)$. Single-product ownership alone does not impose this assumption. With one inside CES product and a fixed outside alternative, demand and expenditure share are
\[
q(p)=\frac{Bap^{-\sigma}}{1+ap^{1-\sigma}},
\qquad
w=\frac{ap^{1-\sigma}}{1+ap^{1-\sigma}}.
\]
The firm's own elasticity magnitude is $E=\sigma-(\sigma-1)w$, not $\sigma$. The first-order condition gives $c=p(E-1)/E$. Since
\[
\frac{dw}{d\log p}=-(\sigma-1)w(1-w),
\qquad
\frac{d\log c}{d\log p}
=1+\frac{(\sigma-1)^2w(1-w)}{E(E-1)}
=\frac{\sigma}{E},
\]
exact percentage pass-through is
\begin{equation}
\label{eq:finite_ces_percentage}
\frac{d\log p}{d\log c}
=\frac{E}{\sigma}
=1-\frac{\sigma-1}{\sigma}w.
\end{equation}
It is below one at a positive inside share and tends to one as that share vanishes. A cost-induced price increase lowers the firm's share, making residual demand more elastic and reducing the percentage markup. For example, $\sigma=2$, $B=2$, $a=1$, $p=1$, and $c=1/3$ give $w=1/2$, unit-cost pass-through $9/4$, and percentage pass-through $3/4$. In a market with several firms, this scalar calculation holds rivals' prices fixed; full equilibrium requires the matrix response. Figure \ref{fig:v2_ces} illustrates the scalar formula.

\paragraph{Local consumer-welfare accounting.}
Let $R_j=p_jq_j$. Write $dCS$ for the consumer-welfare differential in baseline monetary units. Under quasilinear preferences this is the usual consumer-surplus differential. With income effects, including the standard expenditure-based CES interpretation, it is the first-order money-metric welfare change at the baseline; no globally integrable Marshallian consumer-surplus function is presumed. In either interpretation,
\[
dCS=-\mathbf q^\top d\mathbf p
=-\mathbf R^\top\Psi^\tau d\boldsymbol\tau.
\]
With uniformly bounded revenues along the fixed-dimensional sequences above,
\begin{align}
dCS_{\mathrm{logit}}
&=-\sum_j R_j(1-\mu_j)\,d\tau_j
+O(\|\mathbf s\|\,\|d\boldsymbol\tau\|),
\label{eq:cs_logit_percentage}\\
dCS_{\mathrm{CES}}
&=-\sum_j R_j\,d\tau_j
+O(\|\mathbf w\|\,\|d\boldsymbol\tau\|).
\label{eq:cs_ces_percentage}
\end{align}
These remainders describe error in the local welfare derivative caused by the small-share approximation. Replacing a differential by a finite cost change introduces a separate linearization error. Even with a single inside CES product, the finite-share revenue weight is the factor in \eqref{eq:finite_ces_percentage}, rather than one.

The formulas show what revenue and margin data can support under maintained demand restrictions. They do not establish those restrictions or reveal the share-based approximation error from revenues and margins alone. Moreover, $\mu_j$ is a product's price-minus-marginal-economic-cost margin divided by price. An accounting profit margin, or an average over products priced separately, is not automatically this object. An application must justify the mapping from the observed accounts to products, pricing portfolios, and marginal costs.

\subsection{Local Merger Analysis}
\label{subsec:merger_analysis}

A merger changes the sales that a firm internalizes when setting each price. The resulting incentive shift can be written in the same units as a cost shock. This connection underlies upward-pricing-pressure analysis \citep{farrell2010antitrust} and the multiproduct first-order approach of \citet{jaffe2013first}. The purpose here is to express the distinction between cost and merger responses through the demand and ownership decomposition, and to clarify when its approximations apply.

\paragraph{The incentive change and its exact local response.}
Hold demand and marginal costs fixed, and let $\Omega^{\mathrm{pre}}$ and $\Omega^{\mathrm{post}}$ describe premerger and postmerger ownership. Let $\mathbf p^0$ be a regular premerger equilibrium and $\mathbf m=\mathbf p-\mathbf c$ the vector of absolute markups. The shift in the normalized pricing equations is
\begin{align}
g_j(\mathbf p)
&=f_j(\mathbf p;\Omega^{\mathrm{post}})-f_j(\mathbf p;\Omega^{\mathrm{pre}})\notag\\
&=\sum_{l\ne j}(\Omega^{\mathrm{post}}_{jl}-\Omega^{\mathrm{pre}}_{jl})
(p_l-c_l)D_{j\to l}(\mathbf p),
\label{equation:upward.pricing.pressure}
\end{align}
or, writing $\Lambda^a=\Lambda(\mathbf p;\Omega^a)$,
\begin{equation}
\label{equation:upp.vector}
\mathbf g(\mathbf p)=(\Lambda^{\mathrm{pre}}-\Lambda^{\mathrm{post}})\mathbf m.
\end{equation}
For substitutes and positive markups, newly internalized diversion creates upward pricing pressure on the merging firms' products. Products with unchanged ownership incentives have zero entries in $\mathbf g$, but their equilibrium prices can still respond through the pricing Jacobian.

To distinguish a derivative from a completed merger, introduce an incentive parameter $\eta$:
\[
\mathbf f(\mathbf p,\eta)
=\mathbf f(\mathbf p;\Omega^{\mathrm{pre}})+\eta\mathbf g(\mathbf p).
\]
At $\eta=0$ this is the premerger system, and at $\eta=1$ it is the postmerger system. Between these endpoints it interpolates the newly internalized profit weights. The implicit-function theorem gives the exact derivative at the baseline,
\begin{equation}
\label{equation:merger.price.effects}
\left.\frac{d\mathbf p}{d\eta}\right|_{\eta=0}
=-[\mathbf J^{\mathrm{pre}}(\mathbf p^0)]^{-1}\mathbf g(\mathbf p^0),
\qquad
\mathbf J^a=\frac{\partial\mathbf f(\mathbf p;\Omega^a)}{\partial\mathbf p^\top}.
\end{equation}
If $\Lambda^{\mathrm{pre}}(\mathbf p^0)$ is invertible, define the cost-equivalent vector
\[
\mathbf t^{\mathrm{merger}}
=[\Lambda^{\mathrm{pre}}(\mathbf p^0)]^{-1}\mathbf g(\mathbf p^0).
\]
Then \eqref{equation:merger.price.effects} equals $\Psi^{\mathrm{pre}}\mathbf t^{\mathrm{merger}}$. For premerger single-product firms, $\Lambda^{\mathrm{pre}}=I$ and this vector is simply $\mathbf g$. The normalized incentive shift and its local response remain well defined even when a cost-equivalent vector cannot be recovered by inversion.

\paragraph{Why the postmerger Jacobian differs.}
\citet{jaffe2013first} distinguish the matrix that transmits merger incentives from ordinary premerger cost pass-through and study conditions under which the latter approximates the former. In the normalization used here, one Newton step for the postmerger system at $\mathbf p^0$ is
\begin{equation}
\label{eq:merger_post_newton}
\Delta\mathbf p_{\mathrm N}
=-[\mathbf J^{\mathrm{pre}}(\mathbf p^0)+\mathbf g_p(\mathbf p^0)]^{-1}
\mathbf g(\mathbf p^0),
\qquad
\mathbf g_p=\frac{\partial\mathbf g}{\partial\mathbf p^\top},
\end{equation}
provided the displayed matrix is nonsingular. Since the demand component $\mathbf K$ is invariant to ownership at fixed prices, Lemma \ref{lemma:matrix.curvature.decomposition} implies
\[
\mathbf g_p(\mathbf p)
=\mathbf C(\mathbf p;\Omega^{\mathrm{post}})
-\mathbf C(\mathbf p;\Omega^{\mathrm{pre}}).
\]
Equivalently, its $(j,k)$ entry is
\[
(g_p)_{jk}
=(\Lambda^{\mathrm{pre}}-\Lambda^{\mathrm{post}})_{jk}
+\sum_l m_l\frac{\partial(\Lambda^{\mathrm{pre}}-\Lambda^{\mathrm{post}})_{jl}}{\partial p_k}.
\]
The first term records the change in the margins on newly internalized sales; the second records how diversion changes with prices. This is the ownership part of the difference between the two Jacobians.

The exact derivative in \eqref{equation:merger.price.effects} is invariant to a smooth, nonsingular row rescaling of the pricing equations at their baseline zero. The finite Newton approximation \eqref{eq:merger_post_newton} need not be: premerger prices generally do not solve the postmerger equations, so derivatives of a price-dependent normalization enter its Jacobian. Thus \eqref{eq:merger_post_newton} specifies an approximation in this paper's normalization; it is not an assertion that all formulations of a postmerger Newton step coincide.

\paragraph{Approximation, information, and welfare.}
The approximation results earlier in the paper can simplify the inverse in a local merger calculation when their conditions hold. Logit, for example, has an identity limit for unit-cost pass-through along the fixed-dimensional small-share sequences considered here. Such limits help interpret the simulation evidence on pricing-pressure approximations in \citet{miller2017upward}. They do not identify curvature from slopes alone, or establish that an elasticity-based diagonal correction is more accurate than the identity. A retained feedback term is useful when its approximation error is controlled or its performance is validated in the relevant demand environment.

Three distinct uncertainties remain in an application: measurement of the incentive vector and pricing derivatives, approximation of the local inverse, and extrapolation from the local derivative to the completed merger. A good approximation to the first two objects does not bound the third. If an equilibrium path exists and remains regular over $\eta\in[0,1]$, its exact price change satisfies
\[
\mathbf p(1)-\mathbf p(0)
=-\int_0^1
[\mathbf J^{\mathrm{pre}}(\mathbf p(\eta))+\eta\mathbf g_p(\mathbf p(\eta))]^{-1}
\mathbf g(\mathbf p(\eta))\,d\eta.
\]
The baseline formula replaces this changing integrand by its value at zero. A finite-merger application therefore requires control of that movement or a check against solved counterfactual equilibria.

For a small incentive change, the consumer-welfare differential is $dCS=-\mathbf q^\top d\mathbf p$, under the monetary interpretation given above. With unchanged costs, the change in aggregate operating profit is
\[
dPS=\mathbf q^\top d\mathbf p+\mathbf m^\top\mathbf G\,d\mathbf p,
\qquad
\mathbf G=\frac{\partial\mathbf q}{\partial\mathbf p^\top}.
\]
The quantity-adjustment term cannot generally be omitted. Nor can the envelope theorem for each premerger firm's own pricing choices be applied to the combined profit of the newly merged entity at premerger prices: those prices need not maximize its joint profit. Under quasilinear preferences, total surplus changes by $dW=\mathbf m^\top\mathbf G\,d\mathbf p$ when costs and product characteristics are unchanged. These accounting identities keep the welfare calculation separate from the approximation used to obtain the price response.

\section{Conclusion\label{section:conclusion}}

The same change in pricing incentives can have different market-wide effects depending on how demand and ownership transmit it across products. This paper develops a tractable representation of that transmission under multiproduct Bertrand competition. The normalized pricing system makes own curvature, changes in substitution, and portfolio internalization explicit. The approximation and limiting results then establish when a simpler calculation retains the interactions needed for incidence analysis.

Two findings guide that simplification. Individual profit curvature describes a firm's adjustment with rivals fixed, while equilibrium feedback determines how those adjustments combine. Small unconditional shares likewise leave open whether consumers continue to substitute within a nest or among products purchased by a selected group of consumers. The logit and CES benchmarks yield diagonal limits under the stated experiments, while nested logit and gamma mixing can preserve cross-product responses. In the nested case, a group-average formula retains those spillovers without a product-level matrix inversion.

These results connect the general use of pass-through in first-order analysis to explicit features of empirical demand models. Once an intervention's effect on pricing incentives is specified, the same local response system applies to changes in costs, demand, or ownership. The series remainder and residual bound assess approximation of that response; model-specific limits reveal which demand restrictions support familiar incidence benchmarks. The welfare and merger applications illustrate how these ingredients enter substantive calculations, with the direct effects of an intervention and the extrapolation to a finite change treated separately.

This interpretation helps explain how specification-sensitive demand derivatives enter incidence predictions and how the effects of omitted interactions can be assessed. The framework also identifies an empirical priority. A policy prediction may depend on mixed demand derivatives or incentive shifts that the available variation only weakly disciplines. Establishing which conclusions remain valid under that uncertainty requires additional identification or robustness results. The characterization here provides a basis for that work by specifying the relevant combinations of demand shape, ownership, and equilibrium feedback.

\paragraph{Disclosure of AI Use}

The author used OpenAI's ChatGPT to assist with literature searches, mathematical checks, coding, and drafting and revising the text. The author retains sole responsibility for the paper's content, including its arguments, results, references, and any errors.

\singlespacing
\bibliographystyle{econ}
\bibliography{references}
\doublespacing

\part*{Appendix}
\appendix
\section{Proofs \label{section:proofs}}

The proofs follow the order of the main results. The first three establish the pricing decomposition and distinguish individual profit curvature from equilibrium feedback. The next group establishes the approximation and demand-profile results. The final proofs justify the incidence transformations and their logit and CES limits.

\subsection{Proof of Lemma \ref{lemma:matrix.curvature.decomposition}}
\label{section:proof.of.lemma.2}
Write $q_{j,k}=\partial q_j/\partial p_k$ and $q_{j,kl}=\partial^2q_j/\partial p_k\partial p_l$. The normalized equation is
\[
f_j=-\frac{q_j}{q_{j,j}}-(p_j-c_j)+\sum_{l\ne j}\Omega_{jl}(p_l-c_l)D_{j\to l}.
\]
The quotient rule gives
\[
\partial_{p_j}\!\left(-\frac{q_j}{q_{j,j}}\right)=-1+\frac{q_jq_{j,jj}}{q_{j,j}^2},\qquad
\partial_{p_k}\!\left(-\frac{q_j}{q_{j,j}}\right)=-\frac{q_{j,k}}{q_{j,j}}+\frac{q_jq_{j,jk}}{q_{j,j}^2}\quad(k\ne j).
\]
The own markup contributes another $-1$ on the diagonal. Differentiating the ownership term yields
\[
\mathbf1\{k\ne j\}\Omega_{jk}D_{j\to k}+\sum_{l\ne j}\Omega_{jl}(p_l-c_l)\frac{\partial D_{j\to l}}{\partial p_k}.
\]
Together these derivatives are exactly $-2I+\mathbf K+\mathbf C$. When $q_{j,k}\ne0$, substituting $\delta_{jk}=-q_{j,k}/q_{j,j}$ and $\kappa^j_{kj}=q_jq_{j,kj}/(q_{j,k}q_{j,j})$ gives $K_{jk}=\delta_{jk}(1-\kappa^j_{kj})$. The expanded expression remains defined when the cross slope is zero. \qed

\paragraph{Mixed-logit buyer composition.}
For \eqref{equation:mixture.composition}, suppose $s_{ij}>0$ is twice continuously differentiable on an open price neighborhood. Assume integrable envelopes justify differentiating its integral twice and that $s_{ij}\|\nabla\log s_{ij}\|^2$ is integrable there. Let $\bar s_j=\int s_{ij}\,dF$, $w_{ij}=s_{ij}/\bar s_j$, and $g_{ij}=\nabla_{\mathbf p}\log s_{ij}$. Since $q_j=M\bar s_j$ with fixed $M$,
\[
\nabla_{\mathbf p}\log q_j=\mathbb E_j[g_{ij}],
\qquad
\partial_{p_k}w_{ij}=w_{ij}\bigl(g_{ij,k}-\mathbb E_j[g_{ij,k}]\bigr).
\]
Differentiating the first identity, including the derivative of these weights, yields
\[
\partial_{p_kp_l}\log q_j
=\mathbb E_j[\partial_{p_kp_l}\log s_{ij}]
+\mathbb E_j[g_{ij,k}g_{ij,l}]
-\mathbb E_j[g_{ij,k}]\mathbb E_j[g_{ij,l}],
\]
which proves \eqref{equation:mixture.composition}. Writing $\eta_{jj}=\partial_{p_j}\log q_j<0$, the quotient rule also gives
\[
K_{jj}=1+\frac{\partial_{p_jp_j}\log q_j}{\eta_{jj}^2},
\qquad
K_{jk}=\frac{\partial_{p_jp_k}\log q_j}{\eta_{jj}^2}\quad(k\ne j).
\]
Thus the mixture identity supplies exactly the log-demand derivatives entering the demand component of the pricing Jacobian.

For a specialization, let type $i$'s utility from product $j$ be $v_{ij}-\alpha_i p_j+\epsilon_{ij}$, with $\alpha_i>0$, the outside option's deterministic utility normalized to zero, and independent type-I extreme-value choice shocks. The joint distribution of $(\alpha_i,v_{i1},\ldots,v_{iJ})$ is otherwise unrestricted subject to the stated integrability conditions. Set $a_{ij}=\alpha_i(1-s_{ij})$. Individual logit derivatives give $g_{ij,j}=-a_{ij}$, $g_{ij,k}=\alpha_i s_{ik}$ for $k\ne j$, and
\[
\partial_{p_jp_j}\log s_{ij}=-\alpha_i^2s_{ij}(1-s_{ij}),
\qquad
\partial_{p_jp_k}\log s_{ij}=\alpha_i^2s_{ij}s_{ik}\quad(k\ne j).
\]
Consequently,
\begin{align*}
K_{jj}=\kappa_{jj}^j
&=1+\frac{\operatorname{Var}_j(a_{ij})-\mathbb E_j[\alpha_i^2s_{ij}(1-s_{ij})]}{\mathbb E_j[a_{ij}]^2},\\
K_{jk}
&=\frac{\mathbb E_j[\alpha_i^2s_{ij}s_{ik}]-\operatorname{Cov}_j(a_{ij},\alpha_i s_{ik})}{\mathbb E_j[a_{ij}]^2},\qquad k\ne j.
\end{align*}
The own-price expression balances a nonnegative variance component against individual logit log-concavity. A change in the preference distribution can alter both components; the expression does not assert a monotone effect of every increase in heterogeneity. The cross-price expression shows why own curvature alone does not determine the demand-side interaction terms. These are finite-price derivative identities, not assertions about the size of equilibrium spillovers in a particular estimated model. The general mixture identity also allows nonlinear price utility; its individual Hessian then includes the corresponding utility-curvature terms. \qed

\subsection{Proof of Lemma \ref{lemma:directional.soc}}
\label{section:proof.of.lemma.1}
Along the path $\mathbf p_f+tv$, profits are $\Phi_{f,v}(t)=\sum_{j\in\mathcal J_f}(m_j+tv_j)\phi_{j,v}(t)$. Differentiating twice gives
\[
\Phi''_{f,v}(0)=2\sum_jv_j\phi'_{j,v}(0)+\sum_jm_j\phi''_{j,v}(0)
=-2S_f(v)+M_f(v)\kappa_f(v).
\]
The equality $\sum_jv_j\phi'_{j,v}(0)=v^\top\mathbf J_{\mathbf q_f}v=-S_f(v)$ uses that a quadratic form depends only on the symmetric part. The aggregate definition of $\kappa_f$ applies whenever $M_f(v)>0$, including directions with zero product-level first derivatives. Factoring out $M_f(v)$ yields the identity and the equivalent directional inequality.

If all markups are positive, $M_f(v)=(\mathbf J_{\mathbf q_f}v)^\top\operatorname{diag}(m_j/q_j)(\mathbf J_{\mathbf q_f}v)$. Nonsingularity of $\mathbf J_{\mathbf q_f}$ then gives $M_f(v)>0$ for every nonzero $v$. The equivalence with negative semidefiniteness of the profit Hessian follows by testing its quadratic form in every direction. \qed

\subsection{Proof of Proposition \ref{prop:directional.pass.through}}
\label{section:proof.of.proposition.1}
Let $H_f=\nabla^2_{\mathbf p_f\mathbf p_f^\top}\Pi_f$ and $\mathbf F_f=\nabla_{\mathbf p_f}\Pi_f$. Since $\mathbf f_f=\mathbf W_f^{-1}\mathbf F_f$, differentiation at a stationary point removes the derivative of the normalizing matrix and gives
\[
\mathbf W_f\mathbf J_{\mathbf f_f}=H_f.
\]
In particular, strict negative definiteness of $H_f$ makes the own-firm pricing block invertible. The implicit function theorem gives a differentiable branch of stationary prices with rivals fixed, and continuity of the Hessian makes these prices strict local profit maxima. Premultiplying its fixed-rival response equation by $\hat v_f^\top\mathbf W_f$ and writing $d\mathbf p_f=a\hat v_f$, where $a=\|d\mathbf p_f\|_{\mathbf W_f}>0$, gives
\[
a(-\hat v_f^\top H_f\hat v_f)=\hat v_f^\top\mathbf W_f\Lambda_fd\mathbf t_f.
\]
By Lemma \ref{lemma:directional.soc}, $-\hat v_f^\top H_f\hat v_f=M_f(\hat v_f)[2R_f(\hat v_f)-\kappa_f(\hat v_f)]>0$. Division proves \eqref{equation:directional.norm.identity}. \qed

\paragraph{A globally defined example of strategic feedback.}
Let $0\leq\beta<1$. A quasilinear consumer chooses nonnegative quantities to maximize $u_\beta(\mathbf q)-\mathbf p^\top\mathbf q$, where
\[
u_\beta(q_1,q_2)
=\frac{3-2\beta}{1-\beta}(q_1+q_2)
-\frac{q_1^2+2\beta q_1q_2+q_2^2}{2(1-\beta^2)}.
\]
The utility Hessian is negative definite, and the objective tends to minus infinity as $\|\mathbf q\|\to\infty$. Consumer demand is therefore uniquely defined at every nonnegative price vector. When both quantities are positive, the consumer's first-order conditions give
\[
q_i=3-2\beta-p_i+\beta p_k,\qquad k\ne i.
\]
At boundaries, demand is determined by the nonnegativity constraints in the consumer's problem; the affine expressions are not extended to negative quantities. In a neighborhood of $(p_1,p_2)=(2,2)$, both goods are purchased, and demand is smooth, downward sloping, and substitutable.

Each firm sells one product, has marginal cost one, and can choose any nonnegative price. Fix its rival's price at $p_k=2$. For $1\leq p_i\leq3$, demand for its own good is $3-p_i$, while its rival's demand is $1+\beta(p_i-2)>0$. At $p_i\geq3$, its own demand is zero. Thus, for prices at or above cost, its profit is
\[
\pi_i(p_i,2)=
\begin{cases}
(p_i-1)(3-p_i)=1-(p_i-2)^2,&1\leq p_i\leq3,\\
0,&p_i\geq3.
\end{cases}
\]
Prices below cost yield nonpositive profit. Hence $p_i=2$ is the unique globally optimal response to $p_k=2$, proving that $(2,2)$ is a Nash equilibrium over unrestricted nonnegative prices. Baseline quantities and markups are one, $q_{i,i}=-1$, $q_{i,ii}=0$, and $\pi_{i,ii}=-2$.

Within the positive-demand region, the normalized pricing condition is
\[
f_i=3-2\beta+c_i-2p_i+\beta p_k=0.
\]
The Jacobian has diagonal $-2$ and off-diagonal $\beta$, with eigenvalues $-2\pm\beta<0$. Differentiation gives
\[
\Psi=\frac{1}{4-\beta^2}\begin{pmatrix}2&\beta\\ \beta&2\end{pmatrix}.
\]
These derivatives describe a branch of Nash equilibria, not merely stationary points. To see this, at the baseline both quantities are strictly positive when firm $i$ prices at its cost and the rival charges two: they are $2$ and $1-\beta$. This remains true for costs and rival prices sufficiently near the baseline. Starting at cost, an increase in $p_i$ lowers its own demand and raises rival demand until its own quantity reaches zero. Profit on this interval is a strictly concave quadratic; above the choke price it is zero, and below cost it is nonpositive. The nearby solution to the pricing conditions is therefore the unique global best response to the nearby rival price. The matrix formula accordingly gives equilibrium cost derivatives.

For $\beta=3/4$, a finite common-cost experiment can also be verified directly. Set $c_i=1+t$ for $t\in[0,1]$ and $p_i^*=2+4t/5$. At these prices, each quantity and markup is $1-t/5>0$. Holding the rival at $p_k^*$, demand evaluated at $p_i=c_i$ is $2-2t/5$ for firm $i$ and $1/4-t/20$ for its rival, both positive. For $p_i\geq c_i$, own demand therefore remains affine until it reaches zero, while rival demand stays positive. The unique maximizer of the resulting profit quadratic is
\[
p_i=\frac{(3/2)+(3/4)p_k^*+c_i}{2}
=2+\frac45t.
\]
It yields positive profit and dominates every below-cost or zero-demand deviation. This proves the asserted Nash equilibrium and exact four-fifths common-cost pass-through throughout $t\in[0,1]$.

\subsection{Proof of Proposition \ref{prop:neumann.structural}}
\label{section:proof.of.proposition.2}
The factorization $\mathbf J_{\mathbf f}=(I-\Gamma)\mathbf A$, with $\Gamma=-\mathbf B\mathbf A^{-1}$, gives $\Psi=-\mathbf A^{-1}(I-\Gamma)^{-1}\Lambda$. Since $\rho(\Gamma)<1$, the Neumann series for $(I-\Gamma)^{-1}$ converges. Subtracting the terms with powers $0$ through $K$ leaves
\[
\sum_{m=K+1}^{\infty}\Gamma^m=\Gamma^{K+1}(I-\Gamma)^{-1},
\]
which proves the exact remainder, preserving the order of $\mathbf A^{-1}$ and $\Lambda$. If $\|\Gamma\|\leq\eta<1$, submultiplicativity and the geometric series bound this tail by $\eta^{K+1}/(1-\eta)$. Multiplying by the bounds on $\mathbf A^{-1}$ and $\Lambda$, or on $\Lambda h$ for a fixed direction, proves the stated inequalities. \qed

\subsection{Proof of Lemma \ref{lemma:semi.elasticity.representation} \label{section:proof.of.lemma.3}}
Since $\eta_{jj}(\mathbf p)=\frac{\partial_j q_j(\mathbf p)}{q_j(\mathbf p)}$, differentiating with respect to \(p_j\) gives
\[
\partial_j\eta_{jj}(\mathbf p)
=
\frac{q_j(\mathbf p)\partial_{jj}q_j(\mathbf p)-(\partial_j q_j(\mathbf p))^2}{q_j(\mathbf p)^2}.
\]
Dividing by $\eta_{jj}(\mathbf p)^2=\frac{(\partial_j q_j(\mathbf p))^2}{q_j(\mathbf p)^2}$ yields
\[
\frac{\partial_j\eta_{jj}(\mathbf p)}{\eta_{jj}(\mathbf p)^2}
=
\frac{q_j(\mathbf p)\partial_{jj}q_j(\mathbf p)}{(\partial_j q_j(\mathbf p))^2}-1
=
\kappa_{jj}^j(\mathbf p)-1,
\]
which gives the first identity.

Next, for \(k\neq j\), $\eta_{jk}(\mathbf p)=\frac{\partial_k q_j(\mathbf p)}{q_j(\mathbf p)}$, so
\[
\partial_j\eta_{jk}(\mathbf p)
=
\frac{q_j(\mathbf p)\partial_{kj}q_j(\mathbf p)-\partial_k q_j(\mathbf p)\partial_j q_j(\mathbf p)}{q_j(\mathbf p)^2}.
\]
Dividing by \(\eta_{jj}(\mathbf p)^2=(\partial_j q_j(\mathbf p))^2/q_j(\mathbf p)^2\) gives
\[
\frac{\partial_j\eta_{jk}(\mathbf p)}{\eta_{jj}(\mathbf p)^2}
=
\frac{q_j(\mathbf p)\partial_{kj}q_j(\mathbf p)}{(\partial_j q_j(\mathbf p))^2}
-\frac{\partial_k q_j(\mathbf p)}{\partial_j q_j(\mathbf p)}.
\]
Using $\delta_{jk}(\mathbf p) = -\frac{\partial_k q_j(\mathbf p)}{\partial_j q_j(\mathbf p)}$ and $\delta_{jk}(\mathbf p)\kappa_{kj}^j(\mathbf p) = -\frac{q_j(\mathbf p)\partial_{kj}q_j(\mathbf p)}{(\partial_j q_j(\mathbf p))^2}$, we obtain
\[
\frac{\partial_j\eta_{jk}(\mathbf p)}{\eta_{jj}(\mathbf p)^2}
=
-\delta_{jk}(\mathbf p)\kappa_{kj}^j(\mathbf p)+\delta_{jk}(\mathbf p)
=
\delta_{jk}(\mathbf p)\bigl(1-\kappa_{kj}^j(\mathbf p)\bigr),
\]
as claimed. If a cross slope vanishes, the preceding combined expression remains well defined and provides the stated continuous extension.
\qed

\subsection{Normalized Demand Profiles: Statement and Proof}
\label{section:proof.of.proposition.3}

Use the normalized price path $\mathbf p_{j,n}(t)=\mathbf p^{(n)}+(t/\lambda_{j,n})e_j$, where $\lambda_{j,n}=-\eta_{jj}(\mathbf p^{(n)})>0$, and define
\[
E_{j,n}(t)=\frac{\eta_{jj}(\mathbf p_{j,n}(t))}{\eta_{jj}(\mathbf p^{(n)})},
\qquad
G_{jk,n}(t)=\frac{\eta_{jk}(\mathbf p_{j,n}(t))}{\lambda_{j,n}},
\qquad
Q_{j,n}(t)=\frac{q_j(\mathbf p_{j,n}(t))}{q_j(\mathbf p^{(n)})}.
\]
The following conditional benchmark supplies the demand-profile statement used in Section~\ref{section:small.share.approximation}.

\begin{proposition}[Constant-curvature normalized profiles]
\label{prop:local.shape.jacobian}
Let \(I\) be a common open interval containing zero. Suppose that, eventually on every compact subinterval of \(I\), demand along \(\mathbf p_{j,n}(t)\) is positive, twice continuously differentiable, and has a negative own-price derivative. Assume, uniformly on those compact subintervals,
\[
-\frac{E'_{j,n}(t)}{E_{j,n}(t)^2}\to a_j,
\qquad
\frac{G'_{jk,n}(t)}{E_{j,n}(t)^2}\to b_{jk}\quad(k\ne j),
\]
where the limits are constants, and suppose \(G_{jk,n}(0)\to\rho_{jk}\). Then:
\begin{enumerate}[(i)]
\item The baseline normalized Jacobian coefficients converge to \(-1+a_j\) and \(b_{jk}\), respectively.
\item On every compact subinterval of \(I\cap\{t:1+a_jt>0\}\),
\[
E_{j,n}(t)\to\frac{1}{1+a_jt},\qquad
G_{jk,n}(t)\to
\begin{cases}
\rho_{jk}+b_{jk}t,&a_j=0,\\[0.3em]
\rho_{jk}+\dfrac{b_{jk}}{a_j}\left(1-\dfrac{1}{1+a_jt}\right),&a_j\ne0,
\end{cases}
\]
uniformly.
\item On the same compact subintervals,
\[
Q_{j,n}(t)\to Q_j(t):=
\begin{cases}
e^{-t},&a_j=0,\\
(1+a_jt)^{-1/a_j},&a_j\ne0,
\end{cases}
\]
uniformly.
\end{enumerate}
\end{proposition}

\noindent\emph{Proof.}

\paragraph{Part (i).}
By Lemma \ref{lemma:semi.elasticity.representation}, $-2+\kappa_{jj}^j
=
-1+\frac{\partial_{p_j}\eta_{jj}}{\eta_{jj}^2},$ and $\delta_{jk}(1-\kappa_{kj}^j)
=
\frac{\partial_{p_j}\eta_{jk}}{\eta_{jj}^2}.$
From the definitions of \(E_{j,n}\) and \(G_{jk,n}\),
\[
-\frac{E'_{j,n}(t)}{E_{j,n}(t)^2}
=
\frac{
\partial_{p_j}\eta_{jj}\!\left(p_j^{(n)}+t/\lambda_{j,n},\,\mathbf p_{-j}^{(n)}\right)
}{
\eta_{jj}\!\left(p_j^{(n)}+t/\lambda_{j,n},\,\mathbf p_{-j}^{(n)}\right)^2
},
\]
and, for each \(k\neq j\),
\[
\frac{G'_{jk,n}(t)}{E_{j,n}(t)^2}
=
\frac{
\partial_{p_j}\eta_{jk}\!\left(p_j^{(n)}+t/\lambda_{j,n},\,\mathbf p_{-j}^{(n)}\right)
}{
\eta_{jj}\!\left(p_j^{(n)}+t/\lambda_{j,n},\,\mathbf p_{-j}^{(n)}\right)^2
}.
\]
Evaluating at \(t=0\) and using the assumed convergence therefore gives
\[
\frac{\partial_{p_j}\eta_{jj}(\mathbf p^{(n)})}{\eta_{jj}(\mathbf p^{(n)})^2}\to a_j,
\qquad
\frac{\partial_{p_j}\eta_{jk}(\mathbf p^{(n)})}{\eta_{jj}(\mathbf p^{(n)})^2}\to b_{jk}.
\]
Substituting into the identities above yields
\[
-2+\kappa_{jj}^j(\mathbf p^{(n)})\to -1+a_j,
\qquad
\delta_{jk}(\mathbf p^{(n)})\bigl(1-\kappa_{kj}^j(\mathbf p^{(n)})\bigr)\to b_{jk}.
\]

\paragraph{Part (ii).}
Fix a compact subinterval \(K\subset I\cap\{t:1+a_j t>0\}\). Enlarge it, if necessary, to the compact interval joining \(K\) to zero; this interval remains inside the same domain. All uniform bounds below are taken on that enlarged interval. Define
\[
c_{j,n}(t):=-\frac{E'_{j,n}(t)}{E_{j,n}(t)^2}.
\]
By assumption, \(c_{j,n}\to a_j\) uniformly on \(K\). Since \(E_{j,n}(0)=1\),
\[
\left(\frac{1}{E_{j,n}(t)}\right)'=c_{j,n}(t),
\qquad
\frac{1}{E_{j,n}(0)}=1,
\]
so
\[
\frac{1}{E_{j,n}(t)}
=
1+\int_0^t c_{j,n}(s)\,ds.
\]
Uniform convergence of \(c_{j,n}\) implies
\[
\sup_{t\in K}\left|\frac{1}{E_{j,n}(t)}-(1+a_j t)\right|\to 0.
\]
Because \(1+a_j t\) is bounded away from zero on \(K\), inversion is continuous, and hence
\[
E_{j,n}(t)\to E_j(t):=
\begin{cases}
1, & a_j=0,\\[0.4em]
\dfrac{1}{1+a_j t}, & a_j\neq 0,
\end{cases}
\]
uniformly on \(K\).

Next fix \(k\neq j\), and define
\[
d_{jk,n}(t):=\frac{G'_{jk,n}(t)}{E_{j,n}(t)^2}.
\]
By assumption, \(d_{jk,n}\to b_{jk}\) uniformly on \(K\). Since \(E_{j,n}\to E_j\) uniformly on \(K\), it follows that
\[
G'_{jk,n}(t)=d_{jk,n}(t)E_{j,n}(t)^2 \to b_{jk}E_j(t)^2
\]
uniformly on \(K\). If \(G_{jk,n}(0)\to \rho_{jk}\), then
\[
G_{jk,n}(t)=G_{jk,n}(0)+\int_0^t G'_{jk,n}(s)\,ds
\]
implies \(G_{jk,n}\to G_{jk}\) uniformly on \(K\), where
\[
G_{jk}(t)=\rho_{jk}+\int_0^t b_{jk}E_j(s)^2\,ds.
\]
Evaluating the integral gives
\[
G_{jk}(t)=
\begin{cases}
\rho_{jk}+b_{jk}t, & a_j=0,\\[0.6em]
\rho_{jk}+\dfrac{b_{jk}}{a_j}\!\left(1-\dfrac{1}{1+a_j t}\right), & a_j\neq 0.
\end{cases}
\]

\paragraph{Part (iii).}
From the definition of \(Q_{j,n}\),
\[
\frac{Q'_{j,n}(t)}{Q_{j,n}(t)}
=
\frac{1}{\lambda_{j,n}}
\eta_{jj}\!\left(p_j^{(n)}+t/\lambda_{j,n},\,\mathbf p_{-j}^{(n)}\right)
=
-E_{j,n}(t),
\qquad
Q_{j,n}(0)=1.
\]
Hence
\[
Q_{j,n}(t)=\exp\!\left(-\int_0^t E_{j,n}(s)\,ds\right).
\]
Since \(E_{j,n}\to E_j\) uniformly on \(K\), we obtain
\[
Q_{j,n}(t)\to Q_j(t):=\exp\!\left(-\int_0^t E_j(s)\,ds\right)
\]
uniformly on \(K\). Substituting the expression for \(E_j\) yields
\[
Q_j(t)=
\begin{cases}
e^{-t}, & a_j=0,\\[0.4em]
(1+a_j t)^{-1/a_j}, & a_j\neq 0.
\end{cases}
\]
This completes the proof.
\qed

\subsection{Proof of Lemma \ref{lem:percentage_pt} \label{section:proof.of.lemma.4}}

Because \(c_j(\tau_j)=c_j e^{\tau_j}\), we have \(d\mathbf c=\widehat{\mathbf c}\,d\boldsymbol\tau\) at \(\boldsymbol\tau=\mathbf 0. \) Since an additive unit tax shock and an additive marginal-cost shock enter the pricing system in the same way, the unit-cost pass-through matrix also gives \(d\mathbf p=\Psi\, d\mathbf c.\) Therefore, \(d\log \mathbf p
=
\widehat{\mathbf p}^{-1}d\mathbf p
=
\widehat{\mathbf p}^{-1}\Psi \widehat{\mathbf c}\, d\boldsymbol\tau,\)
so \(\Psi^\tau=\widehat{\mathbf p}^{-1}\Psi \widehat{\mathbf c}.\) \qed

\subsection{Proof of Lemma \ref{lem:unit_pt_logit_ces}}
\label{section:proof.of.lemma.5}
All bounds below concern a fixed finite number of products and the equilibrium sequences in the statement. Recall $\Psi=-\mathbf J_{\mathbf f}^{-1}\Lambda$.

\paragraph{Logit.}
Diversion satisfies $D_{j\to l}=s_l/(1-s_j)$, so $\Lambda=I+O(\|s\|)$. At equilibrium, Example \ref{example:logit.neumann.convergence} gives diagonal pricing derivatives $-1/(1-s_j)$, zero off-diagonal derivatives within each firm, and $S_fs_k/[(1-s_j)(1-S_f)]$ across firms. Since $S_f=O(\|s\|)$, it follows that $\mathbf J_{\mathbf f}=-I+O(\|s\|)$. Inversion around $-I$ yields $\Psi=I+O(\|s\|)$.

\paragraph{CES.}
Let $E_j=\sigma-(\sigma-1)w_j$, with $\sigma>1$. The normalized pricing equations and diversion ratios are
\[
f_j=\frac{p_j}{E_j}-(p_j-c_j)+\sum_{l\ne j}\Omega_{jl}(p_l-c_l)\frac{(\sigma-1)w_l}{E_j}\frac{p_j}{p_l},
\qquad
D_{j\to l}=\frac{(\sigma-1)w_l}{E_j}\frac{p_j}{p_l}.
\]
Bounded relative prices imply $\Lambda=I+O(\|w\|)$. Share derivatives are
\[
\partial_{p_k}w_l=\frac{(1-\sigma)w_l}{p_k}\bigl(\mathbf1\{l=k\}-w_k\bigr).
\]
Differentiate $f_j$ keeping costs fixed. The own derivative of $p_j/E_j-(p_j-c_j)$ is $-(\sigma-1)/\sigma+O(\|w\|)$ and its off-diagonal derivatives are $O(\|w\|)$. Each derivative of the ownership sum is also $O(\|w\|)$: every term contains an expenditure share, while $(p_l-c_l)/p_l$ and relative prices are bounded. Thus
\[
\mathbf J_{\mathbf f}=-\frac{\sigma-1}{\sigma}I+O(\|w\|).
\]
The diagonal remainder need not be $O(w_j)$ alone, because the firm's other product shares enter it. Inversion around the nonsingular scalar matrix gives
\[
\Psi=\frac{\sigma}{\sigma-1}I+O(\|w\|).
\]
Both inversions are valid for sufficiently small shares, with constants depending on the fixed demand parameters and the stated bounds. \qed

\subsection{Proof of Proposition \ref{prop:percentage_pt_logit_ces}}
\label{section:proof.of.proposition.4}
Lemma \ref{lem:percentage_pt} gives $\Psi^\tau=\widehat{\mathbf p}^{-1}\Psi\widehat{\mathbf c}$. If a remainder matrix $E$ is $O(\epsilon)$, then its scaled entry is
\[
(\widehat{\mathbf p}^{-1}E\widehat{\mathbf c})_{jk}
=E_{jk}\frac{c_k}{p_k}\frac{p_k}{p_j}=O(\epsilon),
\]
by the bounds on cost-to-price and relative-price ratios. Absolute prices and costs need not remain bounded.

Under logit, Lemma \ref{lem:unit_pt_logit_ces} gives $\Psi=I+O(\|s\|)$. Therefore, writing $\mu_j=(p_j-c_j)/p_j$,
\[
\Psi^\tau_{\mathrm{logit}}=\operatorname{diag}(c_j/p_j)+O(\|s\|)
=\operatorname{diag}(1-\mu_j)+O(\|s\|).
\]
Under CES the same lemma gives $\Psi=\sigma/(\sigma-1)I+O(\|w\|)$. Dividing the equilibrium pricing equation by $p_j$ yields
\[
\mu_j=\frac1{\sigma-(\sigma-1)w_j}
+\sum_{l\ne j}\Omega_{jl}\mu_l\frac{(\sigma-1)w_l}{\sigma-(\sigma-1)w_j}.
\]
Bounded relative margins and fixed dimension imply $\mu_j=1/\sigma+O(\|w\|)$, hence $c_j/p_j=(\sigma-1)/\sigma+O(\|w\|)$. Combining this with the scaled remainder bound gives
\[
\Psi^\tau_{\mathrm{CES}}=I+O(\|w\|).
\]
This is a small-share limit for finite-market CES, not an exact identity at positive shares. \qed

\section{Proofs for the Model-Specific Small-Share Limits}
\label{appendix:model.specific.small.share.limits}

This appendix verifies the entries in Table \ref{table:small.share.limits} and the ownership restrictions needed to interpret them as pass-through. The ray calculations isolate selection among consumers who continue purchasing; the boundary calculations retain local cross slopes; the nested-logit calculation retains substitution within a group. Together they explain why superficially similar small shares can imply different equilibrium responses. Throughout, recall the identities
\begin{equation}
\label{eq:appendix.identity.diag}
-2+\kappa_{jj}^j(\mathbf p)
=
-1+\frac{\partial_{p_j}\eta_{jj}(\mathbf p)}{\eta_{jj}(\mathbf p)^2},
\end{equation}
and, for \(k\neq j\),
\begin{equation}
\label{eq:appendix.identity.offdiag}
\delta_{jk}(\mathbf p)\bigl(1-\kappa_{kj}^j(\mathbf p)\bigr)
=
\frac{\partial_{p_j}\eta_{jk}(\mathbf p)}{\eta_{jj}(\mathbf p)^2}.
\end{equation}
Thus, for each model, it suffices to characterize the limiting behavior of the normalized derivatives of own and cross semi-elasticities.

It is useful to distinguish between two kinds of small-share limits. For logit, CES, and mixed logit, the natural limit is along a \emph{ray} $\mathbf p^{(n)}=\lambda_n \mathbf p$, $\lambda_n\to\infty$, with fixed \(\mathbf p\in \mathbb R_{++}^J\), so relative prices remain fixed and all inside shares vanish. By contrast, for linear and AIDS demand, the natural small-share limit is a \emph{boundary limit} in which product \(j\)'s quantity in the linear model, or expenditure share in the Stone specification, tends to zero while remaining nonnegative; these models are not naturally analyzed by sending all prices to infinity along a common ray.

\subsection{Ray Limits: Logit and CES}

The first set of examples considers ray limits of the form $\mathbf p^{(n)}=\lambda_n \mathbf p$, $\lambda_n\to\infty$, for a fixed vector \(\mathbf p\in\mathbb R_{++}^J\). Along such a sequence, relative prices remain fixed while all prices diverge proportionally. For logit and CES demand, this implies that inside shares vanish and yields a well-defined small-share limit.

\begin{proposition}[Ray limits with fixed relative prices]
\label{prop:ray.limits.standard.models}
Fix \(\mathbf p\in \mathbb R_{++}^J\) and let \(\mathbf p^{(n)}=\lambda_n \mathbf p\) with \(\lambda_n\to\infty\).

\begin{enumerate}[(i)]
\item \textbf{Logit.}
Suppose
\[
q_j(\mathbf p)=
M\frac{\exp(\delta_j-\alpha p_j)}
{1+\sum_{\ell=1}^J \exp(\delta_\ell-\alpha p_\ell)},
\qquad \alpha>0.
\]
Then, for each \(j\),
\[
-2+\kappa_{jj}^j(\mathbf p^{(n)})\to -1,
\qquad
\delta_{jk}(\mathbf p^{(n)})\bigl(1-\kappa_{kj}^j(\mathbf p^{(n)})\bigr)\to 0
\quad (k\neq j).
\]

\item \textbf{CES.}
Suppose
\[
q_j(\mathbf p)=
\frac{B}{p_j}
\frac{\exp\bigl(\delta_j-(\sigma-1)\log p_j\bigr)}
{1+\sum_{\ell=1}^J \exp\bigl(\delta_\ell-(\sigma-1)\log p_\ell\bigr)},
\qquad \sigma>1.
\]
Then, for each \(j\),
\[
-2+\kappa_{jj}^j(\mathbf p^{(n)})\to -1+\frac{1}{\sigma},
\qquad
\delta_{jk}(\mathbf p^{(n)})\bigl(1-\kappa_{kj}^j(\mathbf p^{(n)})\bigr)\to 0
\quad (k\neq j).
\]
\end{enumerate}
\end{proposition}
\begin{proof}
\textbf{(i) Logit.}
Let
\[
s_j(\mathbf p)=
\frac{\exp(\delta_j-\alpha p_j)}
{1+\sum_{\ell=1}^J \exp(\delta_\ell-\alpha p_\ell)},
\qquad
q_j(\mathbf p)=Ms_j(\mathbf p).
\]
Then
\[
\eta_{jj}(\mathbf p)=-\alpha(1-s_j(\mathbf p)),
\qquad
\eta_{jk}(\mathbf p)=\alpha s_k(\mathbf p)\quad (k\neq j),
\]
so
\[
\partial_{p_j}\eta_{jj}(\mathbf p)=-\alpha^2 s_j(\mathbf p)(1-s_j(\mathbf p)),
\qquad
\partial_{p_j}\eta_{jk}(\mathbf p)=\alpha^2 s_j(\mathbf p)s_k(\mathbf p).
\]
Hence
\[
\frac{\partial_{p_j}\eta_{jj}(\mathbf p)}{\eta_{jj}(\mathbf p)^2}
=
-\frac{s_j(\mathbf p)}{1-s_j(\mathbf p)},
\qquad
\frac{\partial_{p_j}\eta_{jk}(\mathbf p)}{\eta_{jj}(\mathbf p)^2}
=
\frac{s_j(\mathbf p)s_k(\mathbf p)}{(1-s_j(\mathbf p))^2}.
\]
Along \(\mathbf p^{(n)}=\lambda_n\mathbf p\) with \(\lambda_n\to\infty\), \(s_j(\mathbf p^{(n)})\to 0\) for every \(j\), so both expressions converge to \(0\). Applying \eqref{eq:appendix.identity.diag}--\eqref{eq:appendix.identity.offdiag} yields the claim.

\medskip
\noindent
\textbf{(ii) CES.}
Write
\[
x_j(\mathbf p)=e^{\delta_j}p_j^{-(\sigma-1)},
\qquad
b_j(\mathbf p)=\frac{x_j(\mathbf p)}{1+\sum_{\ell=1}^J x_\ell(\mathbf p)},
\qquad
q_j(\mathbf p)=\frac{B}{p_j}b_j(\mathbf p).
\]
Then
\[
\eta_{jj}(\mathbf p)=\frac{-\sigma+(\sigma-1)b_j(\mathbf p)}{p_j},
\qquad
\eta_{jk}(\mathbf p)=\frac{\sigma-1}{p_k}b_k(\mathbf p)\quad (k\neq j).
\]
Moreover,
\[
\partial_{p_j}b_j(\mathbf p)=-(\sigma-1)\frac{b_j(\mathbf p)(1-b_j(\mathbf p))}{p_j},
\qquad
\partial_{p_j}b_k(\mathbf p)=(\sigma-1)\frac{b_j(\mathbf p)b_k(\mathbf p)}{p_j}\quad (k\neq j),
\]
so a direct calculation gives
\[
\frac{\partial_{p_j}\eta_{jj}(\mathbf p)}{\eta_{jj}(\mathbf p)^2}
=
\frac{
\sigma-(\sigma-1)b_j(\mathbf p)-(\sigma-1)^2b_j(\mathbf p)(1-b_j(\mathbf p))
}{
[-\sigma+(\sigma-1)b_j(\mathbf p)]^2
},
\]
and, for \(k\neq j\),
\[
\frac{\partial_{p_j}\eta_{jk}(\mathbf p)}{\eta_{jj}(\mathbf p)^2}
=
\frac{(\sigma-1)^2}{[-\sigma+(\sigma-1)b_j(\mathbf p)]^2}
\frac{p_j}{p_k}b_j(\mathbf p)b_k(\mathbf p).
\]
Along \(\mathbf p^{(n)}=\lambda_n\mathbf p\), we have \(b_j(\mathbf p^{(n)})\to 0\) for every \(j\), while \(p_j^{(n)}/p_k^{(n)}=p_j/p_k\) remains constant. Therefore
\[
\frac{\partial_{p_j}\eta_{jj}(\mathbf p^{(n)})}{\eta_{jj}(\mathbf p^{(n)})^2}\to \frac{1}{\sigma},
\qquad
\frac{\partial_{p_j}\eta_{jk}(\mathbf p^{(n)})}{\eta_{jj}(\mathbf p^{(n)})^2}\to 0.
\]
Applying \eqref{eq:appendix.identity.diag}--\eqref{eq:appendix.identity.offdiag} completes the proof.
\end{proof}

\paragraph{Multiproduct ownership along the same rays.}
Suppose these price vectors are interior equilibria with markups \(m_\ell^{(n)}=O(\lambda_n)\). For logit,
\[
D_{j\to\ell}=\frac{s_\ell}{1-s_j}\quad(\ell\ne j),
\]
and its price derivatives are bounded by a constant times \(D_{j\to\ell}\). Since every share decays exponentially along a positive ray, \(D_{j\to\ell}\to0\) and \(m_\ell\partial_kD_{j\to\ell}\to0\). For CES,
\[
D_{j\to\ell}
=\frac{(\sigma-1)(p_j/p_\ell)b_\ell}{\sigma-(\sigma-1)b_j}.
\]
Fixed relative prices and \(b_\ell\to0\) imply \(D_{j\to\ell}\to0\). Differentiating this expression gives \(\partial_kD_{j\to\ell}=O(b_\ell/\lambda_n)\), so its markup-weighted derivatives also vanish. Consequently the limiting unit pass-through matrices are \(I\) for logit and \(\sigma/(\sigma-1)I\) for CES, for any fixed ownership partition. These limits do not assert exact diagonal pass-through at finite shares.

\subsection{Ray Limits: Mixed Logit}

The lower tail of the random price coefficient determines which consumers remain relevant as prices rise. Lognormal and gamma mixing provide contrasting examples under the same fixed-product, fixed-utility experiment.

\begin{proposition}[Ray limits for mixed logit]
\label{prop:ray.limits.mixed.logit}
Fix finite mean utilities \(\delta_j\), a fixed finite product set, and \(\mathbf p\in\mathbb R_{++}^J\). Let \(\mathbf p^{(n)}=\lambda_n\mathbf p\), where \(\lambda_n\to\infty\), and consider
\[
q_j(\mathbf p)=M\int
\frac{e^{\delta_j-\alpha p_j}}{1+\sum_\ell e^{\delta_\ell-\alpha p_\ell}}\,dF(\alpha).
\]
\begin{enumerate}[(a)]
\item If \(\log\alpha\sim N(\mu,v)\), with \(v>0\), the demand-side diagonal entries converge to \(-1\) and the off-diagonal entries converge to zero.
\item If \(\alpha\sim\mathrm{Gamma}(r,\beta)\), with shape \(r>0\) and rate \(\beta>0\), define
\[
\Psi_j(z;\mathbf p)=\frac{e^{\delta_j-zp_j}}{1+\sum_\ell e^{\delta_\ell-zp_\ell}},
\quad H_j(\mathbf p)=\int_0^\infty z^{r-1}\Psi_j(z;\mathbf p)\,dz,
\quad \ell_j=-\partial_j\log H_j(\mathbf p)>0.
\]
Then
\[
-2+\kappa_{jj}^j(\lambda_n\mathbf p)\to-1+\frac{\partial_{jj}\log H_j}{\ell_j^2},
\qquad
\delta_{jk}(1-\kappa_{kj}^j)(\lambda_n\mathbf p)\to\frac{\partial_{jk}\log H_j}{\ell_j^2}
\quad(k\ne j).
\]
The limits depend on the ray and generally retain interactions between products. The own curvature correction has no universal sign.
\end{enumerate}
\end{proposition}

\begin{proof}
Write \(s_j(\alpha,\mathbf p)=e^{\delta_j-\alpha p_j}/(1+\sum_\ell e^{\delta_\ell-\alpha p_\ell})\).

\paragraph{(a) Lognormal mixing.}
For nonnegative integers \(m\), let \(L_m(x)=\mathbb E[\alpha^m e^{-x\alpha}]\). The moment scales, rather than only the demand-level tail, are needed to compare derivatives. Completing the square in \(y=\log\alpha\) gives
\[
L_m(x)=\frac{e^{m\mu+m^2v/2}}{\sqrt{2\pi v}}
\int_{\mathbb R}\exp\!\left[-xe^y-\frac{(y-\mu-mv)^2}{2v}\right]dy.
\]
Let \(w_m=W(vxe^{\mu+mv})\), where \(W(z)e^{W(z)}=z\). The exponent is minimized at \(y_m=\mu+mv-w_m\), where \(xe^{y_m}=w_m/v\). Expanding about this minimum on the scale \(\sqrt{v/(1+w_m)}\) yields
\begin{equation}
\label{eq:lognormal.moment.laplace}
\log L_m(x)=m\mu+\frac{m^2v}{2}
-\frac{w_m^2+2w_m}{2v}-\frac12\log(1+w_m)+o(1).
\end{equation}
For completeness, after the change \(y=y_m+u\sqrt{v/(1+w_m)}\), the quadratic term is \(u^2/2\) and the cubic remainder on each bounded \(u\)-interval is \(O(w_m^{-1/2})\). Convexity of the exponent bounds the two tails outside an expanding neighborhood of the minimum. The rescaled integral therefore converges to \(\sqrt{2\pi}\), giving \eqref{eq:lognormal.moment.laplace}.

For clarity, the exponent's exact increase about its minimum, writing $w=w_m$ and $h=y-y_m$, is
\[
D(h)=\frac{w}{v}(e^h-1-h)+\frac{h^2}{2v}.
\]
Set $h=a u$, with $a=\sqrt{v/(1+w)}$. Then $D(au)\to u^2/2$ uniformly on bounded intervals. Convexity and the limiting derivatives at $u=1$ and $u=-1$ give uniform exponential tail bounds. These bounds remain integrable after multiplication by $e^{kau}$ for fixed $k=1,2$ and sufficiently large $w$, justifying convergence of the first two tilted moments as well as the integral itself.

The same concentration about \(y_m\), including its first two exponential moments, implies
\begin{equation}
\label{eq:lognormal.moment.ratios}
\frac{L_m(cx)}{L_m(x)}=x^{-\log(c)/v+o(1)}\to0\quad(c>1),
\qquad
\frac{L_{m+1}(x)}{L_m(x)}\sim\frac{w_m}{vx}=O(\log x/x),
\qquad
\frac{L_0(x)L_2(x)}{L_1(x)^2}\to1.
\end{equation}
In the first expression the decay is a fixed power, up to a subpower factor; no faster-than-every-power bound is used.

Set \(x=\lambda_n p_j\). Expanding \(1/(1+\sum_\ell e^{\delta_\ell-\alpha\lambda_np_\ell})\) about one bounds the demand error by a finite sum of \(L_0(\lambda_n(p_j+p_\ell))\). Differentiating once or twice gives the same bounds with \(L_1\) or \(L_2\). The first limit in \eqref{eq:lognormal.moment.ratios} therefore gives
\[
q_j\sim Me^{\delta_j}L_0(x),\qquad
-\partial_jq_j\sim Me^{\delta_j}L_1(x),\qquad
\partial_{jj}q_j\sim Me^{\delta_j}L_2(x).
\]
Thus \(\kappa_{jj}^j\to1\).

For \(k\ne j\), set \(c=(p_j+p_k)/p_j>1\). Product-share bounds inside the integrals give
\[
|\partial_kq_j|\le C L_1(cx),\qquad
|\partial_{jk}q_j|\le C L_2(cx).
\]
The quotient rule and \(\eta_{jj}^2\sim[L_1(x)/L_0(x)]^2\) then bound the normalized cross derivative by
\[
C\left\{\frac{L_2(cx)L_0(x)}{L_1(x)^2}
+\frac{L_1(cx)}{L_1(x)}\right\}\longrightarrow0,
\]
using \eqref{eq:lognormal.moment.ratios}. Lemma \ref{lemma:semi.elasticity.representation} proves part (a).

\paragraph{(b) Gamma mixing.}
Let \(C_r=M\beta^r/\Gamma(r)\). The substitution \(z=\lambda_n\alpha\) gives
\[
q_j(\lambda_n\mathbf p)=C_r\lambda_n^{-r}
\int_0^\infty z^{r-1}e^{-\beta z/\lambda_n}\Psi_j(z;\mathbf p)\,dz.
\]
The kernel and its first two price derivatives are bounded in a neighborhood of the fixed positive ray by constants times \((1+z^2)e^{-az}\), for some \(a>0\). Dominated convergence therefore applies to the demand and derivative integrals. Remembering that derivatives with respect to the actual price vector introduce one factor \(\lambda_n^{-1}\) each, we obtain
\[
q_j(\lambda_n\mathbf p)\sim C_r\lambda_n^{-r}H_j(\mathbf p),
\]
\[
\partial_kq_j(\lambda_n\mathbf p)=C_r\lambda_n^{-r-1}\partial_kH_j+o(\lambda_n^{-r-1}),
\quad
\partial_{jk}q_j(\lambda_n\mathbf p)=C_r\lambda_n^{-r-2}\partial_{jk}H_j+o(\lambda_n^{-r-2}).
\]
Consequently \(\eta_{jj}=-\lambda_n^{-1}\ell_j+o(\lambda_n^{-1})\) and
\[
\partial_j\eta_{jk}(\lambda_n\mathbf p)
=\lambda_n^{-2}\partial_{jk}\log H_j+o(\lambda_n^{-2}).
\]
Here \(\ell_j>0\) follows by differentiating the strictly decreasing own-price kernel under the integral. Substitution in the semi-elasticity identities proves part (b).
\end{proof}

\paragraph{What common-ray regular variation does and does not imply.}
For gamma mixing, the selected consumers have \(\alpha\) of order \(1/\lambda_n\), so their conditional probabilities of rival products need not vanish. This accounts for the dependence on \(H_j(\mathbf p)\). Equivalently,
\[
H_j(\mathbf p)=e^{\delta_j}\Gamma(r)p_j^{-r}C_j(\mathbf p),\qquad
C_j(\mathbf p)=\frac{1}{\Gamma(r)}\int_0^\infty
\frac{u^{r-1}e^{-u}}{1+\sum_\ell e^{\delta_\ell-(p_\ell/p_j)u}}\,du.
\]
With no rival inside goods, \(C_j\) is constant and the own correction is \(1/r\). With rivals, derivatives of \(C_j\) can overturn the sign of this correction. For instance, in the symmetric two-product example in the main text, differentiating the one-dimensional integral gives \(\partial_{jj}\log H_j/\ell_j^2\simeq-0.188\).

The same dominated-convergence argument, uniformly on compact positive-price neighborhoods, gives the normalized own-price profile
\[
Q_j(t)=\frac{H_j(\mathbf p+t e_j/\ell_j)}{H_j(\mathbf p)},\qquad t>-\ell_jp_j.
\]
This profile changes relative prices, whereas the common-ray experiment holds them fixed. It need not have constant curvature. These are demand-derivative results for every \(r>0\); equilibrium incidence additionally requires an admissible interior equilibrium and a nonsingular limiting pricing Jacobian. For example, the one-product limit has a zero pricing-Jacobian limit at \(r=1\), so a finite inverse limit cannot be inferred there.

\paragraph{Lognormal diversion and ownership terms.}
Suppose, in addition, that the ray is a sequence of interior equilibria and markups are \(O(\lambda_n)\). Let \(x=\lambda_n p_j\), \(c=(p_j+p_\ell)/p_j>1\), and \(\ell\ne j\). The preceding moment bounds imply
\[
D_{j\to\ell}=O\!\left(\frac{L_1(cx)}{L_1(x)}\right)\to0.
\]
The quotient rule, applied directly to \(-\partial_jq_\ell/\partial_jq_j\), gives for every \(k\)
\[
|\partial_kD_{j\to\ell}|
\le C\left\{\frac{L_2(cx)}{L_1(x)}
+\frac{L_1(cx)L_2(x)}{L_1(x)^2}\right\}.
\]
Since \(L_2(x)/L_1(x)=O(\log x/x)\), it follows that
\[
\lambda_n|\partial_kD_{j\to\ell}|
\le O(\log\lambda_n)\frac{L_1(cx)}{L_1(x)}\longrightarrow0.
\]
Thus every markup-weighted diversion derivative vanishes. For a fixed ownership partition, \(\mathbf C\to0\) and \(\Lambda\to I\); combined with part (a), this gives \(\mathbf P\to I\). The result establishes a limit, not a finite-share accuracy ranking between diagonal approximations. Convergence is pointwise in the fixed ray and mixing parameters and can be slow: the decay rate in \eqref{eq:lognormal.moment.ratios} depends on relative prices and mixing dispersion. The identity limit therefore need not be an accurate approximation at empirically relevant shares.

\subsection{Boundary Limits: Linear and AIDS}

The remaining examples are most naturally analyzed along feasible boundary sequences: quantity demand tends to zero in the linear model, while expenditure share tends to zero in the Stone specification. A vanishing expenditure share also implies vanishing quantity if absolute prices have finite positive limits. The proposition states the weaker conditions sufficient for the normalized-derivative limits.

\begin{proposition}[Boundary limits for linear and AIDS demand]
\label{prop:boundary.limits.linear.aids}
Fix product \(j\).

\begin{enumerate}[(i)]
\item \textbf{Linear demand.}
Suppose
\[
q_j(\mathbf p)=\gamma_j-\sum_{\ell=1}^J \beta_{j\ell}p_\ell,
\qquad
\beta_{jj}>0.
\]
Let \(\{\mathbf p^{(n)}\}_{n\ge 1}\) be any sequence in the domain \(\{q_j>0\}\) such that $q_j(\mathbf p^{(n)})\downarrow 0.$ Then
\[
-2+\kappa_{jj}^j(\mathbf p^{(n)})\to -2,
\qquad
\delta_{jk}(\mathbf p^{(n)})\bigl(1-\kappa_{kj}^j(\mathbf p^{(n)})\bigr)\to -\frac{\beta_{jk}}{\beta_{jj}}
\quad (k\neq j).
\]
In fact, these expressions are constant on \(\{q_j>0\}\), so the limits do not depend on the particular boundary sequence.

\item \textbf{LA/AIDS with Stone index.}
Suppose
\[
q_j(\mathbf p)=\frac{B}{p_j}w_j(\mathbf p),
\qquad
w_j(\mathbf p)
=
\alpha_j+\sum_{\ell=1}^J \gamma_{j\ell}\log p_\ell+\beta_j\log\!\Big(\frac{B}{P^S(\mathbf p)}\Big),
\]
where the Stone price index is
\[
\log P^S(\mathbf p)=\sum_{\ell=1}^J \omega_\ell \log p_\ell,
\qquad
\sum_{\ell=1}^J \omega_\ell = 1,
\]
with fixed weights \(\omega_\ell\). Let \(\{\mathbf p^{(n)}\}_{n\ge 1}\) be a sequence such that
\[
w_j(\mathbf p^{(n)})\downarrow 0,
\qquad
\frac{p_j^{(n)}}{p_k^{(n)}}\to \overline c_{jk}\in(0,\infty)
\quad (k\neq j),
\]
and assume
\[
\gamma_{jj}-\beta_j\omega_j<0.
\]
Then
\[
-2+\kappa_{jj}^j(\mathbf p^{(n)})\to -2,
\]
and, for each \(k\neq j\),
\[
\delta_{jk}(\mathbf p^{(n)})\bigl(1-\kappa_{kj}^j(\mathbf p^{(n)})\bigr)
\to
-\overline c_{jk}\frac{\gamma_{jk}-\beta_j\omega_k}{\gamma_{jj}-\beta_j\omega_j}.
\]
\end{enumerate}
\end{proposition}

\begin{proof}
\textbf{(i) Linear demand.}
Since \(q_j(\mathbf p)=\gamma_j-\sum_{\ell=1}^J \beta_{j\ell}p_\ell,\) we have $\eta_{jj}(\mathbf p)
=
-\frac{\beta_{jj}}{q_j(\mathbf p)}$, and $\eta_{jk}(\mathbf p)
=
-\frac{\beta_{jk}}{q_j(\mathbf p)}$ for $k \neq j$. Differentiating with respect to \(p_j\), $\partial_{p_j}\eta_{jj}(\mathbf p)
=
-\frac{\beta_{jj}^2}{q_j(\mathbf p)^2}$ and $\partial_{p_j}\eta_{jk}(\mathbf p)
=
-\frac{\beta_{jk}\beta_{jj}}{q_j(\mathbf p)^2}.$ Therefore $\frac{\partial_{p_j}\eta_{jj}(\mathbf p)}{\eta_{jj}(\mathbf p)^2}
=
-1$ and $\frac{\partial_{p_j}\eta_{jk}(\mathbf p)}{\eta_{jj}(\mathbf p)^2}
=
-\frac{\beta_{jk}}{\beta_{jj}}.$ Applying \eqref{eq:appendix.identity.diag}--\eqref{eq:appendix.identity.offdiag} gives $-2+\kappa_{jj}^j(\mathbf p)=-2$ and $\delta_{jk}(\mathbf p)\bigl(1-\kappa_{kj}^j(\mathbf p)\bigr)
=
-\frac{\beta_{jk}}{\beta_{jj}},$ for every \(\mathbf p\) with \(q_j(\mathbf p)>0\). In particular, the same limits hold along any sequence with \(q_j(\mathbf p^{(n)})\downarrow 0\).

\medskip
\noindent
\textbf{(ii) LA/AIDS with Stone index.}
Define $A_{jk}:=\gamma_{jk}-\beta_j\omega_k.$ Since $q_j(\mathbf p)=\frac{B}{p_j}w_j(\mathbf p),$ we have $\log q_j(\mathbf p)=\log B-\log p_j+\log w_j(\mathbf p),$ and therefore
\[
\eta_{jk}(\mathbf p)
=
-\frac{\mathbf 1\{j=k\}}{p_j}
+
\frac{1}{w_j(\mathbf p)}\partial_{p_k}w_j(\mathbf p).
\]
Because $\log P^S(\mathbf p)=\sum_{\ell=1}^J \omega_\ell \log p_\ell,$ it follows that $\partial_{p_k}w_j(\mathbf p)=\frac{A_{jk}}{p_k}.$ Therefore
\[
\eta_{jj}(\mathbf p)
=
\frac{A_{jj}-w_j(\mathbf p)}{p_j w_j(\mathbf p)},
\qquad
\eta_{jk}(\mathbf p)=\frac{A_{jk}}{p_k w_j(\mathbf p)}
\quad (k\neq j).
\]

Along the stated sequence, \(w_j(\mathbf p^{(n)})\to 0\), so
\[
\eta_{jj}(\mathbf p^{(n)})\sim \frac{A_{jj}}{p_j^{(n)}w_j(\mathbf p^{(n)})},
\qquad
\eta_{jk}(\mathbf p^{(n)})=\frac{A_{jk}}{p_k^{(n)}w_j(\mathbf p^{(n)})}.
\]

Since \(\partial_{p_j}w_j(\mathbf p)=A_{jj}/p_j\), we obtain
\[
\partial_{p_j}\eta_{jj}(\mathbf p)
=
\frac{w_j(\mathbf p)^2-A_{jj}w_j(\mathbf p)-A_{jj}^2}{p_j^2 w_j(\mathbf p)^2},
\]
and hence
\[
\frac{\partial_{p_j}\eta_{jj}(\mathbf p)}{\eta_{jj}(\mathbf p)^2}
=
\frac{w_j(\mathbf p)^2-A_{jj}w_j(\mathbf p)-A_{jj}^2}{(A_{jj}-w_j(\mathbf p))^2}
\to -1.
\]
Applying \eqref{eq:appendix.identity.diag} yields
\[
-2+\kappa_{jj}^j(\mathbf p^{(n)})\to -2.
\]

For \(k\neq j\),
\[
\partial_{p_j}\eta_{jk}(\mathbf p)
=
-\frac{A_{jk}A_{jj}}{p_k p_j w_j(\mathbf p)^2},
\]
so
\[
\frac{\partial_{p_j}\eta_{jk}(\mathbf p)}{\eta_{jj}(\mathbf p)^2}
=
-\frac{p_j}{p_k}\frac{A_{jk}A_{jj}}{(A_{jj}-w_j(\mathbf p))^2}.
\]
Along the sequence, \(w_j(\mathbf p^{(n)})\to 0\) and \(p_j^{(n)}/p_k^{(n)}\to \overline c_{jk}\), so
\[
\frac{\partial_{p_j}\eta_{jk}(\mathbf p^{(n)})}{\eta_{jj}(\mathbf p^{(n)})^2}
\to
-\overline c_{jk}\frac{A_{jk}}{A_{jj}}
=
-\overline c_{jk}\frac{\gamma_{jk}-\beta_j\omega_k}{\gamma_{jj}-\beta_j\omega_j}.
\]
Applying \eqref{eq:appendix.identity.offdiag} yields
\[
\delta_{jk}(\mathbf p^{(n)})\bigl(1-\kappa_{kj}^j(\mathbf p^{(n)})\bigr)
\to
-\overline c_{jk}\frac{\gamma_{jk}-\beta_j\omega_k}{\gamma_{jj}-\beta_j\omega_j}.
\]
This proves the claim.
\end{proof}

\paragraph{Scope and interpretation.}
These are rowwise statements about product \(j\), conditional on a feasible sequence. They neither assert that all product demands can vanish simultaneously nor establish a nonsingular full pricing Jacobian. In the Stone specification, the weights are held fixed. If a finite-price boundary interpretation is desired, additionally require \(p_k^{(n)}\to\bar p_k\in(0,\infty)\); then quantity \(q_j\) also vanishes and the semi-elasticity divergence rates stated below apply.

Linear demand is the simplest boundary case. As product \(j\)'s demand approaches zero, the own semi-elasticity and cross semi-elasticities with nonzero \(\beta_{jk}\) diverge at rate \(q_j^{-1}\), and their own-price derivatives diverge at rate \(q_j^{-2}\). A zero cross coefficient gives identically zero cross terms. After normalization, however, these divergences cancel exactly. The diagonal term converges to \(-2\), while the off-diagonal term converges to \(-\beta_{jk}/\beta_{jj}\). Thus the small-share limit preserves the relative cross-slope structure of the linear demand system. In this sense, linear demand can remain non-diagonal even near the boundary: substitution patterns continue to matter at first order.

Under finite positive limiting prices, LA/AIDS exhibits a similar boundary logic, but with richer dependence on the sequence by which product \(j\)'s expenditure share vanishes. As \(w_j(\mathbf p^{(n)})\downarrow 0\), the own semi-elasticity and cross semi-elasticities with nonzero \(A_{jk}\) diverge like \(w_j^{-1}\), and their own-price derivatives diverge like \(w_j^{-2}\). Cross terms with \(A_{jk}=0\) are identically zero. The normalized diagonal term therefore converges to the same boundary benchmark \(-2\). The off-diagonal term, however, depends not only on the demand primitives \(\gamma_{jk}\) and \(\beta_j\), but also on the limiting relative price ratio \(\overline c_{jk}\) and the Stone-index weights \(\omega_k\). The term \(\gamma_{jk}-\beta_j\omega_k\) is the relevant log-price derivative of \(w_j\) with respect to \(p_k\), so the limit preserves the first-order substitution structure encoded by the LA/AIDS system rather than washing it out.

Taken together, these two examples show that boundary-type demand systems behave very differently from logit-type ray limits. In logit and CES, off-diagonal interaction terms vanish asymptotically, so the limiting Jacobian becomes diagonal. In linear and LA/AIDS demand, by contrast, nonzero cross slopes generate normalized off-diagonal terms that survive. Small shares therefore do not imply asymptotic independence across products: in boundary models, substitution patterns continue to shape pass-through even as product \(j\)'s quantity demand or expenditure share becomes arbitrarily small.

\subsection{Nested Logit with Multiproduct Ownership}
\label{appendix:nested.logit.benchmark}

Let products belong to a fixed finite collection of nests. Under the standard nested-logit normalization, define
\[
v_j=\exp\!\left(\frac{\delta_j-\alpha p_j}{1-\sigma}\right),
\quad V_g=\sum_{j\in\mathcal J_g}v_j,
\quad z_j=\frac{v_j}{V_{g(j)}},
\quad s_g=\frac{V_g^{1-\sigma}}{1+\sum_hV_h^{1-\sigma}},
\quad q_j=Ms_{g(j)}z_j,
\]
where \(\alpha>0\) and \(0\le\sigma<1\) are fixed. The outside alternative has normalized attractiveness one. This specification makes the price scale explicit: \(-q_j/\partial_jq_j\) contains the factor \((1-\sigma)/\alpha\).

\begin{proposition}[Nested-logit benchmark]
\label{prop:nested.logit.small.share}
Fix an ownership partition and consider a sequence of interior equilibria for the demand system above, allowing costs to vary. Suppose \(s_g\to0\) for every nest and \(z_j\to\bar z_j\in[0,1]\) for every product. Write
\[
B_j^*=1-\sigma\bar z_j,
\qquad
S_{fg}=\sum_{\ell\in\mathcal J_f\cap\mathcal J_g}\bar z_\ell.
\]
Then:
\begin{enumerate}[(i)]
\item The demand-side matrix \(-2I+\mathbf K\) has the block-diagonal limit
\[
J_{jk}^{d,*}=
\begin{cases}
-1-\dfrac{\sigma\bar z_j(1-\bar z_j)}{(B_j^*)^2},&k=j,\\[0.6em]
\dfrac{\sigma\bar z_j\bar z_k}{(B_j^*)^2},&k\ne j,\ g(k)=g(j),\\[0.6em]
0,&g(k)\ne g(j).
\end{cases}
\]
Diversion converges to
\[
D_{j\to k}^*=
\begin{cases}
-1,&k=j,\\
\sigma\bar z_k/B_j^*,&k\ne j,\ g(k)=g(j),\\
0,&g(k)\ne g(j),
\end{cases}
\qquad \Lambda^*=-\Omega\circ D^*.
\]
\item Products in firm--nest group \((f,g)\) have the common limiting markup
\begin{equation}
\label{eq:nested.corrected.markup}
m_{fg}^*=\frac{1-\sigma}{\alpha(1-\sigma S_{fg})}.
\end{equation}
In particular, limiting markups lie between \((1-\sigma)/\alpha\) and \(1/\alpha\).
\item Define the matrix \(L^*\) by
\begin{equation}
\label{eq:nested.markup.derivative}
L_{jk}^*=
\begin{cases}
-\dfrac{\sigma\bar z_k\bigl(\mathbf1\{k\in\mathcal J_{f(j)}\}-S_{f(j),g(j)}\bigr)}{(1-\sigma S_{f(j),g(j)})^2},&g(k)=g(j),\\[0.9em]
0,&g(k)\ne g(j).
\end{cases}
\end{equation}
Then \(I-L^*\) is nonsingular and equilibrium unit pass-through satisfies
\begin{equation}
\label{eq:nested.corrected.passthrough}
\mathbf P\longrightarrow\mathbf P^*:=(I-L^*)^{-1}.
\end{equation}
The limit is block diagonal by nest, and each block satisfies \(P_g^*\mathbf1=\mathbf1\). If a firm owns the entire nest, \(P_g^*=I\).
\end{enumerate}
\end{proposition}

\begin{proof}
Set \(B_j=1-\sigma z_j-(1-\sigma)s_j\). Direct differentiation gives
\[
\eta_{jj}=-\frac{\alpha B_j}{1-\sigma},
\qquad
\partial_kz_\ell=-\frac{\alpha}{1-\sigma}z_\ell(\mathbf1\{\ell=k\}-z_k)
\quad(g(k)=g(\ell)),
\]
with \(\partial_kz_\ell=0\) across nests. Own and cross derivatives of \(\eta\), combined with Lemma \ref{lemma:semi.elasticity.representation}, give part (i). For example,
\[
\frac{\partial_j\eta_{jj}}{\eta_{jj}^2}
=-\frac{\sigma z_j(1-z_j)}{B_j^2}-\frac{(1-\sigma)s_j}{B_j},
\]
and the second term vanishes. The exact off-diagonal diversion ratios are
\[
D_{j\to k}=
\begin{cases}
\dfrac{\sigma z_k+(1-\sigma)s_k}{B_j},&g(k)=g(j),\\[0.6em]
\dfrac{(1-\sigma)s_k}{B_j},&g(k)\ne g(j),
\end{cases}
\quad k\ne j.
\]
These converge to the displayed values because \(B_j\to B_j^*\ge1-\sigma>0\).

The normalized first-order conditions are \(\Lambda m=b\), where
\[
b_j=-\frac{q_j}{\partial_jq_j}=\frac{1-\sigma}{\alpha B_j}
\longrightarrow\frac{1-\sigma}{\alpha B_j^*}.
\]
Every row of \(\Lambda^*\) is strictly diagonally dominant: its off-diagonal absolute row sum is at most
\[
\frac{\sigma(1-\bar z_j)}{1-\sigma\bar z_j}<1.
\]
Hence \(\Lambda^*\) is nonsingular, and \(m\to(\Lambda^*)^{-1}b^*\). Multiplying limiting row \(j\) by \(B_j^*\) gives
\[
m_j^*-\sigma\sum_{\ell\in\mathcal J_{f(j)}\cap\mathcal J_{g(j)}}\bar z_\ell m_\ell^*
=\frac{1-\sigma}{\alpha}.
\]
All products in a firm--nest group therefore have the same markup. Solving the resulting scalar equation proves \eqref{eq:nested.corrected.markup}.

For pass-through, write the demand-implied markup field as \(\mathfrak{m}(\mathbf p)=\Lambda(\mathbf p)^{-1}b(\mathbf p)\). In the neighborhood of the sequence, the normalized pricing equations are
\[
\mathbf f=\Lambda(\mathbf p)\{\mathfrak{m}(\mathbf p)-(\mathbf p-\mathbf c-\mathbf t)\}.
\]
At equilibrium, differentiating gives \(J_{\mathbf f}=\Lambda(\partial\mathfrak{m}/\partial\mathbf p^\top-I)\). Entries of \(\Lambda\), \(b\), and their first derivatives are rational functions of \(s\) and \(z\), with denominators bounded away from zero. In addition, \(\partial_ks_j=O(s_j)\). Thus the derivative limit follows by differentiating the limiting markup field, not by an unsupported interchange of a limit and a derivative. Within group \((f,g)\),
\[
\partial_k S_{fg}
=-\frac{\alpha}{1-\sigma}z_k(\mathbf1\{k\in\mathcal J_f\}-S_{fg})
\quad(k\in\mathcal J_g),
\]
which, substituted into \eqref{eq:nested.corrected.markup}, yields \(\partial\mathfrak{m}/\partial\mathbf p^\top\to L^*\).

It remains to verify nonsingularity. If \((I-L_g^*)v=0\), the identical rows of \(L_g^*\) within a firm--nest group imply that \(v_j=u_f\) within that group. For groups with positive \(S_f=S_{fg}\), these common values satisfy
\[
(1+A_f)u_f-A_f\sum_h S_hu_h=0,
\qquad A_f:=\frac{\sigma S_f}{(1-\sigma S_f)^2}\ge0.
\]
The group matrix \(\operatorname{diag}(1+A)-A S^\top\) is strictly diagonally dominant, with dominance margin one. Hence all \(u_f=0\). For a zero-share group, the corresponding rows of \(L_g^*\) are zero, so its coordinates also vanish. Thus \(I-L^*\) is nonsingular. Continuity of inversion now proves \eqref{eq:nested.corrected.passthrough}.

Finally, \(L_g^*\mathbf1=0\) because \(\sum_{k\in\mathcal J_g}\bar z_k=1\) and \(\sum_{k\in\mathcal J_f\cap\mathcal J_g}\bar z_k=S_{fg}\). Therefore \(P_g^*\mathbf1=\mathbf1\). If a single firm owns the nest, \(S_{fg}=1\) and every entry of \(L_g^*\) is zero, giving \(P_g^*=I\).
\end{proof}

The normalization in \eqref{eq:nested.corrected.markup} matters economically. For example, two equally popular products in one nest with \(\sigma=1/2\), \(\alpha=1\), and a common owner have limiting markup one and identity pass-through. Cross-nest effects vanish even when ownership spans nests, because cross-nest diversion and its derivatives vanish while limiting markups remain bounded. Within a nest, split ownership can sustain product-specific spillovers.

\paragraph{Normalized residual demand.}
The same model also separates baseline curvature from a full normalized profile. Fix \(j\), let \(z=\bar z_j\), \(B=1-\sigma z\), and displace only its price by \(t/[-\eta_{jj}(\mathbf p^{(n)})]\). The ratio of the new to the old \(v_j\) tends to \(e^{-t/B}\), while the ratio of the nest sums tends to \(T(t)=1-z+ze^{-t/B}\). Since the outside alternative dominates aggregate nest shares, the resulting demand ratio is
\[
Q_j(t)=e^{-t/B}T(t)^{-\sigma}.
\]
This holds uniformly on bounded \(t\)-intervals and is positive for every finite \(t\). Its conditional share is \(z(t)=ze^{-t/B}/T(t)\), and its normalized curvature is
\[
A_j(t)=-\frac{\sigma z(t)[1-z(t)]}{[1-\sigma z(t)]^2}.
\]
Except in degenerate cases, this function changes with \(t\). Consequently the negative baseline value \(A_j(0)\) does not satisfy the constant-profile hypothesis of Proposition \ref{prop:local.shape.jacobian}, and it does not imply a finite normalized choke point.

\subsection{Proof of Corollary~\ref{corollary:nested.group.response}}
\label{appendix:nested.group.response.proof}

Fix a nest and suppress its index. Let $x=P_g^*h$, $X=\sum_{j\in\mathcal J_g}z_jx_j$, and, for $S_f>0$, let $\overline x_f=S_f^{-1}\sum_{j\in\mathcal J_f\cap\mathcal J_g}z_jx_j$. Define $A_f=\sigma S_f/(1-\sigma S_f)^2$, so $\chi_f=(1+A_f)^{-1}$. All group averages and group sums below concern groups with $S_f>0$. The matrix $L^*$ in Proposition~\ref{prop:nested.logit.small.share} gives
\[
x_j=h_j+A_f(X-\overline x_f).
\]
Averaging within group $f$ yields
\[
\overline x_f=\chi_f\overline h_f+(1-\chi_f)X.
\]
Averaging again across the nest and using $\sum_fS_f=1$ gives $D_gX=\sum_f S_f\chi_f\overline h_f$. Since $\chi_f>0$, $D_g>0$ and $X=H_g$. Substituting back proves \eqref{equation:nested.group.response}. If $S_f=0$, the corresponding row of $L^*$ vanishes, giving $x_j=h_j$ without defining a group average.

For $j$ in a positive-share group $f$, the matrix entries are
\[
(P_g^*)_{jk}=\mathbf1\{j=k\}
 +(1-\chi_f)z_k\left(\frac{\chi_{f(k)}}{D_g}
 -\frac{\mathbf1\{f(k)=f\}}{S_f}\right).
\]
Across firms this expression is nonnegative. Within a firm, $D_g\geq S_f\chi_f$, so off-diagonal entries are nonpositive. On the diagonal,
\[
1\geq(P_g^*)_{jj}
\geq1-(1-\chi_f)\frac{z_j}{S_f}
\geq\chi_f>0.
\]
These inequalities also cover zero-share individual products in a positive-share group. Zero-share groups have identity rows, completing the proof.

\section{Numerical Illustrations\label{section:numerical.illustration}}
\newcommand{\VtwoMeanKOne}{0.323}
\newcommand{\VtwoMaxKOne}{3.139}
\newcommand{\VtwoMaxFOC}{0.000}
\newcommand{\VtwoFiniteMeanMin}{0.179}
\newcommand{\VtwoFiniteMeanMax}{0.444}

Each exercise addresses a distinct use of the theory: how many interaction terms a local prediction needs, how far that derivative extrapolates to a finite cost change, which within-nest responses survive small shares, and how percentage incidence differs from unit incidence. The exercises use simulated demand and cost primitives throughout. Their purpose is to check these mechanisms in transparent benchmarks, rather than estimate an empirical policy effect.

\subsection{Logit Approximation and Finite Cost Changes}

There are six products and two firms, each owning three products. Ownership is fixed, while product characteristics and costs differ. Demand is multinomial logit with outside utility zero and common price coefficient $\alpha=1$. In each of 100 baseline markets, I draw $\delta_j^{(0)}\sim N(1.2,0.4^2)$ and $c_j\sim U[0.2,0.8]$ independently across products. A common utility shift takes ten equally spaced values from $0$ to $-6$, yielding 1,000 configurations. This fixed-dimensional experiment reduces inside shares by making the outside option more attractive relative to all inside goods.

The equilibrium solver uses the common markup within each firm: $m_f(1-S_f)=1/\alpha$, where $S_f$ is the firm's total quantity share. Every accepted solution also satisfies the original product pricing conditions, normalized by own-price demand slopes, with maximum absolute residual below $10^{-10}$. Pass-through is computed analytically and checked independently against differentiated pricing conditions and re-solved equilibria after small product-specific cost perturbations.

Writing the normalized pricing Jacobian as $J_f=A+B$, with diagonal $A$, define $\Gamma=-BA^{-1}$. The Neumann approximations are
\[
\widehat\Psi_K=-A^{-1}\sum_{k=0}^K\Gamma^k\Lambda,\qquad K\in\{0,1,2\}.
\]
Table~\ref{tab:v2_matrix_errors} compares these with the exact local derivative and the identity approximation. The order-1 mean matrix error is $\VtwoMeanKOne\%$, with a maximum of $\VtwoMaxKOne\%$. These results describe this logit calibration, whose substitution structure is particularly tractable.

\begin{table}[H]
\centering
\caption{Accuracy of local pass-through matrix approximations}
\label{tab:v2_matrix_errors}
\begin{tabular}{lrr}
\toprule
Approximation & Mean error (\%) & Maximum error (\%) \\
\midrule
Neumann $K=0$ & 2.285 & 12.375 \\
Neumann $K=1$ & 0.323 & 3.139 \\
Neumann $K=2$ & 0.057 & 0.847 \\
Identity & 8.905 & 28.247 \\
\bottomrule
\end{tabular}

\end{table}

Figure~\ref{fig:neumann_matrix} plots relative Frobenius errors against the feedback norm. Higher-order truncations reduce errors in this exercise; this does not establish a universal ordering of errors across different demand systems.

\begin{figure}[H]
\centering
\includegraphics[width=.98\textwidth]{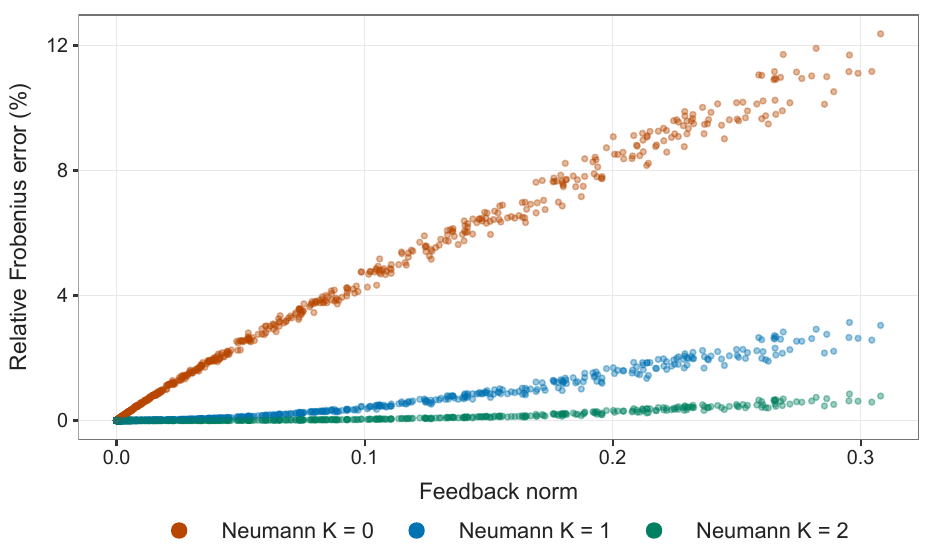}
\caption{Neumann approximation errors across logit configurations}
\label{fig:neumann_matrix}
\begin{minipage}{.95\textwidth}\footnotesize
\emph{Notes:} Each point represents one configuration and truncation order. The vertical axis is $100\|\widehat\Psi_K-\Psi\|_F/\|\Psi\|_F$; the horizontal axis is $\|\Gamma\|_\infty$. No fitted curve is imposed. The reproduction separately checks the convergence norm and the absolute-error remainder bound.
\end{minipage}
\end{figure}

Figure~\ref{fig:price_comparison} considers per-unit cost changes of $0.1$ on all products, firm 1's products, or product 1. It compares averages of the local predictions $\Psi\,d\mathbf t$, $\widehat\Psi_1\,d\mathbf t$, and $d\mathbf t$. The identity tends to overstate average responses as shares rise in this design. Agreement in averages need not imply accuracy for individual products; full response-vector errors are also reported in the reproduction output.

\begin{figure}[H]
\centering
\includegraphics[width=.98\textwidth]{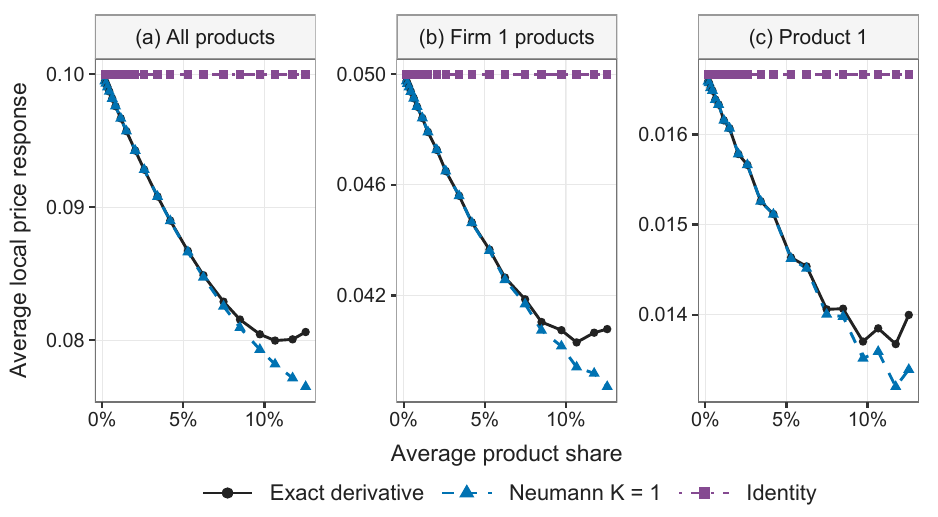}
\caption{Average local price responses under three cost directions}
\label{fig:price_comparison}
\begin{minipage}{.95\textwidth}\footnotesize
\emph{Notes:} Prices are averaged over all six products. Configurations are placed in 20 equal-count bins by average product share, using the same bins for every method and scenario. Lines connect within-bin means. Only $K=1$ enters the Neumann curve; the exact derivative is counted once per configuration and scenario. All curves are local predictions, not re-solved finite-tax equilibria.
\end{minipage}
\end{figure}

A separate exercise re-solves equilibrium after per-unit taxes of $0.01$, $0.1$, and $0.25$ on affected products. Figure~\ref{fig:v2_finite_tax} measures the relative Euclidean error of $\Psi\,d\mathbf t$ against the actual finite equilibrium price change. At tax $0.1$, mean errors across scenarios range from $\VtwoFiniteMeanMin\%$ to $\VtwoFiniteMeanMax\%$. This linearization error is distinct from the error due to approximating the matrix $\Psi$.

\begin{figure}[H]
\centering
\includegraphics[width=.98\textwidth]{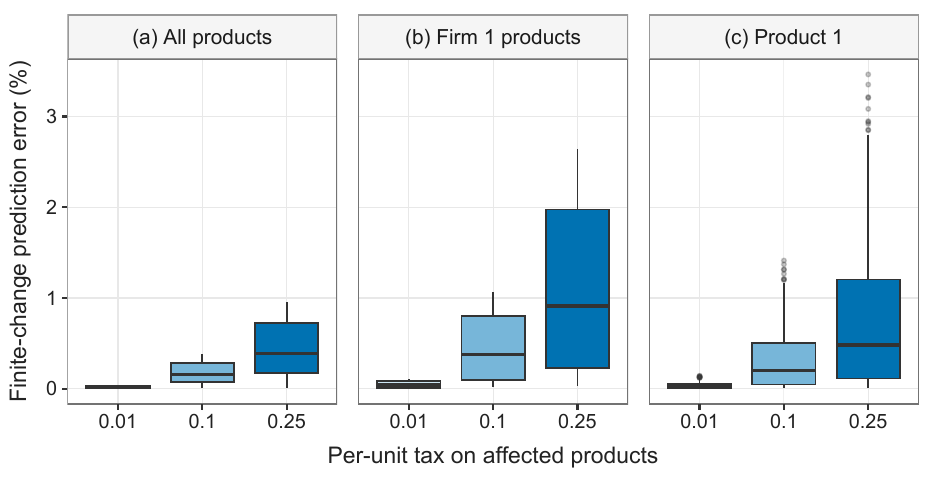}
\caption{Local derivative predictions versus finite-tax equilibria}
\label{fig:v2_finite_tax}
\begin{minipage}{.95\textwidth}\footnotesize
\emph{Notes:} Each box summarizes 1,000 configurations. Boxes indicate medians and interquartile ranges; whiskers extend to observations within 1.5 interquartile ranges. The denominator is the Euclidean norm of the actual price change. All 9,000 post-tax equilibria satisfy normalized pricing conditions to within $10^{-10}$.
\end{minipage}
\end{figure}

\subsection{The Nested-Logit Benchmark}

There are two nests of three products, $\alpha=1$, $\sigma=0.6$, baseline utilities $(1.1,0.6,0.9,1.3,0.8,0.4)$, and costs $(0.4,0.6,0.3,0.5,0.4,0.7)$. A common utility shift varies from $0$ to $-12$ in unit steps. Three ownership structures are compared: single-product firms; portfolios with ownership vector $(1,1,2,1,2,2)$; and one owner per complete nest.

For this nested-logit specification, the sufficiency result of \citet{garrido2024aggregative} implies that a finite solution of the pricing first-order conditions is a Nash equilibrium. The numerical checks assess residuals and local derivatives without assuming equilibrium uniqueness.

As both nest shares vanish, let $z_k$ be product $k$'s limiting conditional share and $S_{fg}$ firm $f$'s conditional share in nest $g$. The limiting common markup within a firm--nest group is
\[
m_{fg}^*=\frac{1-\sigma}{\alpha(1-\sigma S_{fg})}.
\]
I solve these equations jointly with conditional demand. If $L^*$ is the derivative of the limiting markup function with respect to product prices, evaluated at equilibrium, then the limiting pass-through is $(I-L^*)^{-1}$, where
\[
L_{jk}^*=-\frac{\sigma z_k[\mathbf1\{f(j)=f(k)\}-S_{f(j),g(j)}]}{(1-\sigma S_{f(j),g(j)})^2}
\]
when $j$ and $k$ are in the same nest, and zero otherwise. The pass-through matrix $(I-L^*)^{-1}$ has row sums equal to one. Whole-nest ownership yields $L^*=0$ and identity limiting pass-through.

Figure~\ref{fig:v2_nested} shows convergence to this block limit. With single-product firms or portfolios spanning nests, the identity retains nonzero error as unconditional shares vanish. When one firm owns each entire nest, the identity and block limit coincide. These conclusions concern the stated limiting experiment.

\begin{figure}[H]
\centering
\includegraphics[width=.98\textwidth]{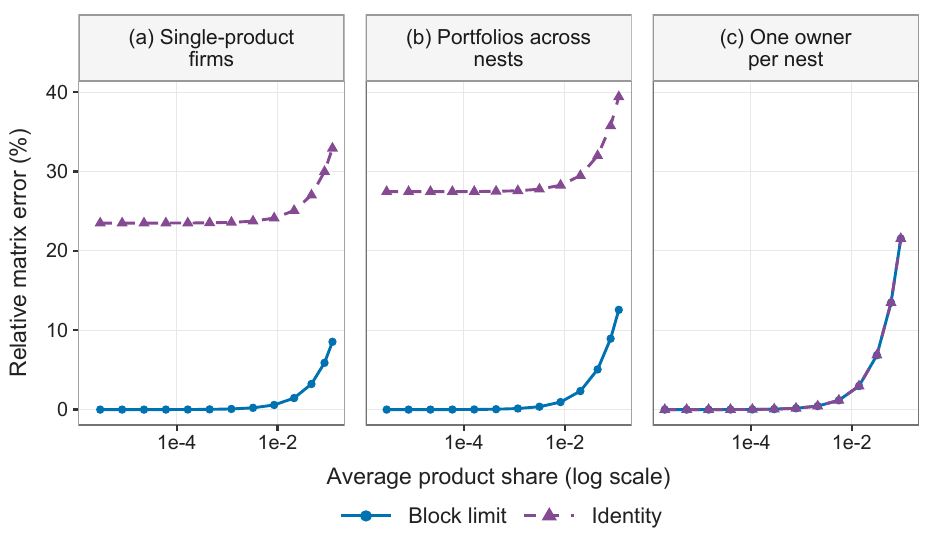}
\caption{Nested-logit limits under alternative ownership structures}
\label{fig:v2_nested}
\begin{minipage}{.95\textwidth}\footnotesize
\emph{Notes:} Relative Frobenius errors compare the finite-market local derivative with the block limit or identity. There are 13 configurations per ownership structure. The horizontal axis is logarithmic. Finite-market solutions satisfy normalized pricing conditions and firm-level second-order conditions. Both curves coincide in the last panel.
\end{minipage}
\end{figure}

\subsection{Finite-Market CES Percentage Pass-Through}

For one inside product and a fixed outside alternative, consider
\[
q(p)=\frac{Bap^{-\sigma}}{1+ap^{1-\sigma}},\qquad
w=\frac{ap^{1-\sigma}}{1+ap^{1-\sigma}},\qquad \sigma>1.
\]
The residual elasticity magnitude is $E=\sigma-(\sigma-1)w$. The monopoly condition $p-c=p/E$ implies
\[
\frac{d\log p}{d\log c}=\frac{E}{\sigma}=1-\frac{\sigma-1}{\sigma}w.
\]
Percentage pass-through is below one at positive finite shares and approaches one as $w\to0$. For example, $\sigma=2$ and $w=1/2$ give $3/4$. Exact unity instead holds for exogenous isoelastic residual demand with a fixed multiplicative demand shifter.

\begin{figure}[H]
\centering
\includegraphics[width=.88\textwidth]{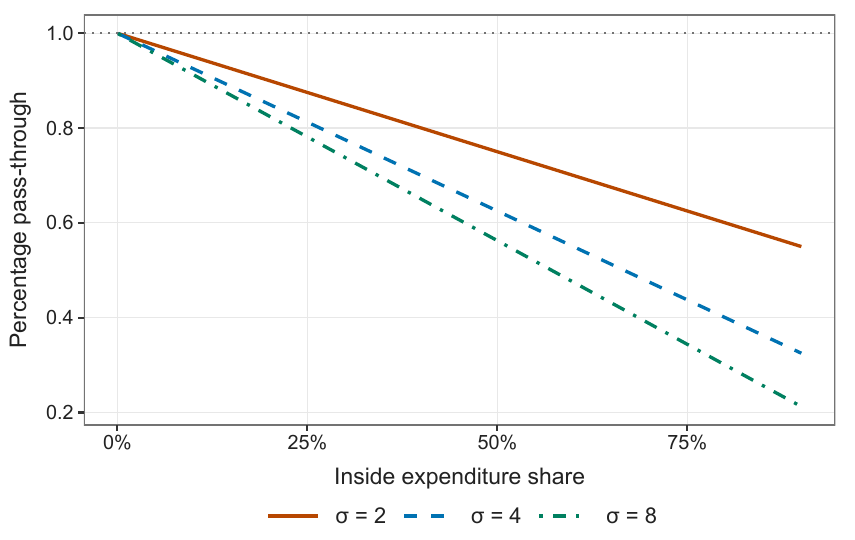}
\caption{Exact percentage pass-through with finite-market CES demand}
\label{fig:v2_ces}
\begin{minipage}{.92\textwidth}\footnotesize
\emph{Notes:} Curves show $1-[(\sigma-1)/\sigma]w$ for $\sigma=2,4,8$. Each point is calibrated to $p=1$ using $a=w/(1-w)$ and $c=1-1/E$. Re-solving after positive and negative log-cost perturbations verifies the derivative at 183 configurations. The dotted line is the small-share limit of one.
\end{minipage}
\end{figure}

\section{Bounds on the Directional Stability Ratio}
\label{app:directional-stability-bounds}

This appendix supports the interpretation of portfolio curvature in Section~\ref{section:characterization.of.pass.through}. It identifies a markup direction that reproduces the scalar benchmark and shows how substitution or complementarity orders the threshold in other directions. These bounds concern an individual firm's second-order condition; equilibrium feedback is governed separately by \eqref{equation:equilibrium.block.feedback}.

The key object is the directional stability ratio
\(R_f(v)\). Let \(m_f=p_f-c_f\) and
\(G_f=\partial q_f(p)/\partial p_f^\top\). Evaluate these objects at an interior stationary point satisfying $q_f+G_f^\top m_f=0$. Suppose that \(q_j>0\) and
\(m_j>0\) for all \(j\in\mathcal J_f\), and that \(G_f\) is symmetric and
negative definite. Define
\[
D_f:=\operatorname{diag}\left(\frac{m_j}{q_j}\right)_{j\in\mathcal J_f},
\qquad
H_f:=-D_f^{1/2}G_fD_f^{1/2}.
\]
Then \(H_f\) is symmetric and positive definite.

\begin{proposition}[Bounds on \(R_f(v)\)]
\label{prop:directional-stability-bounds}
Under the conditions above, for every nonzero direction \(v\),
\[
R_f(v)
=
\frac{z^\top H_f z}{z^\top H_f^2 z},
\qquad
z:=D_f^{-1/2}v,
\]
and hence
\[
\frac{1}{\lambda_{\max}(H_f)}
\leq R_f(v)\leq
\frac{1}{\lambda_{\min}(H_f)}.
\]
Moreover, \(1\) is an eigenvalue of \(H_f\) and \(R_f(m_f)=1\). If the
firm's products are substitutes, so that
\(\partial q_j/\partial p_k\geq0\) for all \(j\neq k\), then
\(\lambda_{\min}(H_f)=1\) and
\[
\frac{1}{\lambda_{\max}(H_f)}
\leq R_f(v)\leq1.
\]
If instead the firm's products are complements, so that
\(\partial q_j/\partial p_k\leq0\) for all \(j\neq k\), then
\(\lambda_{\max}(H_f)=1\) and
\[
1\leq R_f(v)\leq
\frac{1}{\lambda_{\min}(H_f)}.
\]
\end{proposition}

\begin{proof}
Under symmetry,
\[
S_f(v)=-v^\top G_fv,
\qquad
M_f(v)=v^\top G_fD_fG_fv.
\]
Writing \(v=D_f^{1/2}z\) therefore gives
\(S_f(v)=z^\top H_fz\) and \(M_f(v)=z^\top H_f^2z\), which yields the
stated representation of \(R_f(v)\). Expanding \(z\) in an orthonormal
eigenbasis of \(H_f\) gives
\[
R_f(v)
=
\frac{\sum_r\lambda_r z_r^2}
     {\sum_r\lambda_r^2 z_r^2}.
\]
This is a weighted average of \(1/\lambda_r\), with weights proportional
to \(\lambda_r^2z_r^2\), and therefore lies between
\(1/\lambda_{\max}(H_f)\) and \(1/\lambda_{\min}(H_f)\).

At an interior stationary point, symmetry implies
\(q_f+G_fm_f=0\). Since \(m_f=D_fq_f\), letting
\(y=D_f^{1/2}q_f=D_f^{-1/2}m_f\) gives
\[
H_fy=-D_f^{1/2}G_fm_f=D_f^{1/2}q_f=y.
\]
Thus \(1\) is an eigenvalue of \(H_f\), with strictly positive
eigenvector \(y\), and taking \(v=m_f\) gives \(R_f(m_f)=1\).

If products are substitutes, the off-diagonal elements of \(H_f\) are
nonpositive. Choose \(a>\lambda_{\max}(H_f)\). Then
\(K=aI-H_f\) is entrywise nonnegative and satisfies
\(Ky=(a-1)y\). By Perron--Frobenius,
\(\rho(K)=a-1\). Since the eigenvalues of \(K\) are
\(a-\lambda_r(H_f)\), it follows that
\(\lambda_{\min}(H_f)=1\). The stated bound follows.

If products are complements, \(H_f\) is entrywise nonnegative and
\(H_fy=y\) with \(y\gg0\). Perron--Frobenius therefore implies
\(\rho(H_f)=1\). Since \(H_f\) is symmetric and positive definite,
\(\lambda_{\max}(H_f)=1\), which gives the final bound.
\end{proof}

Proposition~\ref{prop:directional-stability-bounds} identifies \(R_f=1\)
as a natural multiproduct analogue of the single-product benchmark.
Under symmetric demand effects, substitution weakly lowers the
directional threshold in the second-order condition,
\(\kappa_f(v)\leq2R_f(v)\), relative to the monopoly bound \(2\),
whereas complementarity weakly raises it.

\end{document}